\documentclass[%
reprint,
 amsmath,amssymb,
 aps,
pra,
]{revtex4-2}
\usepackage{mathrsfs}
\usepackage{graphicx}

\usepackage{dcolumn}
\usepackage{bm}
\usepackage{xcolor}
\usepackage[colorlinks=false]{hyperref}
\hypersetup{
    colorlinks=True,
    linkcolor=blue,
    filecolor=blue,      
    urlcolor=blue,
}
\usepackage{expl3}
\usepackage{subfiles}

\usepackage{amsthm}

\newtheorem{theorem}{Theorem}[section]

\newtheorem{prop}{Proposition}[section]

\newtheorem{cor}[theorem]{Corollary}
\newtheorem{defn}[theorem]{Definition}

\newcommand{\bgd}{\begin{eqnarray}}
\newcommand{\edd}{\end{eqnarray}}
\newcommand{\hf}{\hspace*{0.5cm}}
\newcommand{\pin}{\par\noindent}

\begin{document}

\title{Unveiling topological order through multipartite entanglement}

\author{Siddhartha Patra$^{1}$}
\email{sp14ip022@iiserkol.ac.in}
\author{Somnath Basu$^{2}$}
\email{somnath.basu@iiserkol.ac.in}
\author{Siddhartha Lal$^{1}$}
\email{slal@iiserkol.ac.in}

\affiliation{$^{1}$Department of Physical Sciences, Indian Institute of Science Education and Research-Kolkata, W.B. 741246, India}
\affiliation{$^{2}$Department of Mathematics $\&$ Statistics, Indian Institute of Science Education and Research-Kolkata, W.B. 741246, India}

\date{\today}

\begin{abstract}
\noindent 
It is well known that the topological entanglement entropy ($S_{topo}$) of a topologically ordered ground state in 2 spatial dimensions can be captured efficiently by measuring the tripartite quantum information ($I^{3}$) of a specific annular arrangement of three subsystems. However, the nature of the general N-partite information ($I^{N}$) and quantum correlation of a topologically ordered ground state remains unknown. In this work, we study such $I^N$ measure and its nontrivial dependence on the arrangement of $N$ subsystems. For the collection of subsystems (CSS) forming a closed annular structure, the $I^{N}$ measure ($N\geq 3$) is a topological invariant equal to the product of $S_{topo}$ and the Euler characteristic of the CSS embedded on a planar manifold, $|I^{N}|=\chi S_{topo}$. Importantly, we establish that $I^{N}$ is robust against several deformations of the annular CSS, such as the addition of holes within individual subsystems and handles between nearest-neighbour subsystems.
While the addition of a handle between further neighbour subsystems causes $I^{N}$ to vanish, the multipartite information measures of the two smaller annular CSS emergent from this deformation again yield the same topological invariant. For a general CSS with multiple holes ($n_{h}>1$), we find that the sum of the distinct,  multipartite informations measured on the annular CSS around those holes is given by the product of $S_{topo}$, $\chi$ and $n_{h}$, $\sum^{n_{h}}_{\mu_{i}=1}|I^{N_{\mu_{i}}}_{\mu_{i}}| = n_{h}\chi S_{topo}$. This constrains the concomitant measurement of several multipartite informations on any complicated CSS. The $N^{th}$ order irreducible quantum correlations for an annular CSS of $N$ subsystems is also found to be bounded from above by $|I^{N}|$, which shows the presence of correlations among subsystems arranged in the form of closed loops of all sizes. Thus, our results offer important insight into the nature of the many-particle entanglement and correlations within a topologically ordered state of matter.
\end{abstract}

\maketitle

\section{Introduction}

\par\noindent Topologically ordered \cite{Wen_2013} phases are characterized by nonlocal ordering of matter, and cannot be described by the Ginzburg Landau Wilson paradigm of symmetry broken orders \cite{Wen_2004}. These phases show exotic phenomena like a nontrivial ground state degeneracy observed on a multiply connected spatial manifold
\cite{Tao_Wu_1984,Niu_Wu_1985,Wen_1989,Wen_Zee_1990,Wen_Niu_1990,
Wen_1991_gdeg,Wen_1990}    
and fractional statistics of excitations  \cite{Laughlin_1983,Kitaev_2003,Read_Sachdev_1991,
Wen_1991,Senthil_Matthew_2000,Moessner_Sondhi_2001}. Due to their added topological protection, understanding these phases can lead to novel applications including fault-tolerant quantum computing    \cite{Kitaev_Kong_2012,Sarma_Freedman_2015}. The ground states of several topologically ordered phases have been shown to possess signatures of string-net condensation with long-range entanglement \cite{Levin_Wen_2005,Hamma_Lidar_2008,Lan_Wen_2014,
Fendley_Troyer_2013,Gu_Wen_2009,Slagle_Willianson_2019}. This  refers to the fact that the application of any finite number of local unitary operations on such ground states cannot convert them into direct product form~\cite{chen2010local,gu2009,zeng2019}. Several investigations on the nature of quantum entanglement of various topological phases have been conducted \cite{Blok_Wen_1990,Read_1990,Rokhsar_Kivelson_1988,
Read_Chakraborty_1989,Moessner_Sondhi_2001,Ardonne_Fradkin_2004}. Specifically, the study of the von Neumann entanglement entropy of a singly connected subregion partitioned from a topologically ordered ground state reveals the existence of a geometry independent term that depends on the degeneracy of the ground state manifold. This term is universal in the sense that it is a topological invariant, and is referred to as the topological entanglement entropy (TEE) \cite{Hamma_Zanardi_2005,Levin_Wen_2006,Kitaev_Preskill_2006,Zhang_Vishwanath_2012,
Grover_Vishwanath_2011}. Analysis of the entanglement spectrum of these states has revealed nontrivial degeneracy, a gap that vanishes at topological phase transitions and the existence of gapless edge states \cite{Li_Haldane_2008,Qi_Andreas_2012,Pollmann_Oshikawa_2010,Yao_Qi_2010,LiuFan2011,
CalabreseLefevre2008,LaunchilHaque2010,Schliemann2011,Halperin1982,QiLudwig2012,Laflorencie2016}.

\par\noindent We now discuss the TEE in further detail. It was shown in Refs.\cite{Hamma_Zanardi_2005,Levin_Wen_2006,Kitaev_Preskill_2006} that the entanglement entropy ($S_{A}$) for a subsystem $A$ obtained from a real-space bipartitioning of a  state with nontrivial topological order in 2 spatial dimensions follows the area law \cite{Eisert_Plenio_2010} with a correction term ($\gamma_{A}$), $S_A = \alpha L_A -\gamma_A + ~.~.~.~$ The terms represented by the ellipsis vanish for large subsystem size. Further, it is known that the term $\gamma_A$ is universal: in the simplest setting, it depends on a topological quantum number of the ground state manifold called the quantum dimension ($\mathcal{D}$), and the topology of the subsystem (i.e., the number of disjoint boundary components of the subsystem $A$). The quantum dimension governs the rate of growth of the topologically protected ground state Hilbert space on manifolds with a nontrivial genus. As an example, for the case of the topologically ordered abelian fractional quantum Hall fluids, the quantum dimension $\mathcal{D}=\sqrt{|\text{det} K|}$, where $K$ is the $K$-matrix describing the topological Chern-Simons quantum field theory for these phases. The quantity $|\text{det} K|$ is a count of the number of degenerate ground states on a torus. 
\par\noindent
In general, the von Neumann entanglement measure captures both local as well as nonlocal quantum correlations. Thus, it was shown in Refs.\cite{Levin_Wen_2006,Kitaev_Preskill_2006} that in order to find a purely topological piece of the entanglement entropy, one has to properly choose an arrangement of the subsystems under consideration as well as the corresponding entanglement measure. Specifically, Ref.\cite{Kitaev_Preskill_2006,Levin_Wen_2006} showed that the tripartite information, $I_3(A,B,C)=S_{A}+S_{B}+S_{C}-(S_{AB}+S_{BC}+S_{CA})+S_{ABC}$ is proportional to the TEE ($S_{topo}\equiv\log \mathcal{D}$) for a very specific annular arrangement of three subsystems $A, B$ and $C$: $I_3(A,B,C) = - 2S_{topo}$. Importantly, the $I_3(A, B, C)$ for this particular collection of subsystems (CSS) is defined such that all geometry (or boundary length) dependent terms exactly cancel one another. It has also been shown in Ref.\cite{Hamma_Zanardi_2005} that the entanglement entropy of a singly connected subsystem $A$ is given by $S(A)=(n_A-C_A)\log\mathcal{D}$, where $n_A$ and $C_A$ are the number of links on the perimeter and the number of disjoint boundary components respectively of subsystem $A$. Recent investigations \cite{Liu_Zhou_2016,Kato_Murao_2016} also conclude that the tripartite information measure employed in Refs.\cite{Kitaev_Preskill_2006,Levin_Wen_2006} provide evidence for the presence of nonlocal quantum correlations in topologically ordered ground states in the form of an entanglement Hamiltonian that has tripartite irreducible correlations.
\par\noindent However, certain key questions remain unanswered. First, what is the precise dependence of the tripartite information measurement protocol proposed in Refs.\cite{Kitaev_Preskill_2006,Levin_Wen_2006} on the topology of the collection of subsystems (CSS)? How robust are such measurements against deformations of the CSS topology? Further, given that a topologically ordered system possesses truly long-ranged entanglement~\cite{chen2010local,gu2009,zeng2019}, can the multipartite quantum information be generalised beyond the choice of three subsystems so as to capture unambiguously the TEE? If this can be achieved, what insight does it offer on the nature of multipartite quantum correlations encoded within a topologically ordered state? Answering these questions is the main goal of our work. We summarise our main results below, as well as present a plan of the work.

\subsection{Summary of main results}
\par\noindent In Section \ref{section:IN_definition}, we define a multipartite information measure ($I^N_{\{\mathcal{A}_N\}}$) for a CSS defined by $\{\mathcal{A}_N\}$ (with $N$ number of subsystems in it). $I^N_{\{\mathcal{A}_N\}}$ is a generalization of the tripartite information used to compute the TEE in Refs.\cite{Kitaev_Preskill_2006,Levin_Wen_2006}, and we show that $I^N_{\{\mathcal{A}_N\}}$ is independent of CSS geometry. We then show in Section \ref{section:I_N_computation} that $I^N_{\{\mathcal{A}_N\}}$ is a topological invariant, depending only on the ground state quantum dimension and the Euler characteristic ($\chi$) of the CSS embedded on the underlying planar spatial manifold. Note that $\chi$ is also the classical Euler characteristic of the underlying compactified planar manifold $\mathbb{R}^2$. Specifically, in subsection \ref{bm:simple_annular}, we show that for an annular arrangement of $N\geq 3$ subsystems, $|I^N|=\chi S_{topo}$ (eq.\eqref{eq:IN_annular}).
In the remainder of Section \ref{section:I_N_computation}, we test the robustness of this result against various kinds of deformations of the annular CSS. For instance, in subsection \ref{section:multiple_disconnected_components}, we show that neither the addition of self-loops and holes within subsystems, nor the addition of handles between neighbouring subsystems (subsection \ref{subsec:nnhandles}), changes the result $|I^N|=\chi S_{topo}$. Further, in subsection \ref{section:morethan_nearest_neighbour}, we show that while adding handles between subsystems that are not neighbours causes $I^{N}$ to vanish, the multipartite informations of several smaller annular CSS becomes non-zero. These results are summarised pictorially in Fig.\ref{fig:summary}. Thus, these results establish that the nontrivial topology of an annular CSS is essential for a multipartite information to capture the TEE. Further, it appears very generally possible to identify an annular CSS configuration that is appropriate for such a measurement. 
\par\noindent 
In Section \ref{section:measurement_constraint}, we demonstrate the constraint that governs various multipartite information measurements that can be made in a CSS with $n_{h}$ number of holes (and where a given multipartite information is computed around one of the holes). For instance, we find that the sum of the two multipartite informations of a CSS with $n_{h}=2$, where each is computed individually around one the two holes, adds up to a constant which depends on the product of $S_{topo}$, $n_{h}$ and $\chi$ (the Euler characteristic of the CSS embedded on the underlying  planar manifold). We note that a similar constraint was obtained in Ref.\cite{Zhang_Vishwanath_2012} for a CSS with $n_{h}=2$ placed on the toroidal manifold, and was viewed as an uncertainty relation. We have generalised the constraint to the case of a CSS with $n_{h}\in\mathbb{Z}$ number of holes (eq.\eqref{eq:alltopo}).
 
\par\noindent Finally, in Section \ref{section:TEE_irreducible_correlations}, we study the irreducible quantum correlation content \cite{Liden_Wootters_2002,Zhou_2008,Kim_2021,Zhou_You_2006} encoded within the multipartite information measure of an annular CSS of $N$ subsystems. For a $N$-partite state, the $k-$party irreducible quantum correlation measures that part of the total not arising from any order of correlations less than $k$. We obtain a generalisation of the 3-subsystem strong sub-additivity relation to the case of $N$ subsystems of a topologically ordered state within an annular configuration (eq.\eqref{eq:strong_sub_add_main}). Using this inequality, we show that the $N$-party irreducible correlation is bounded from above by $S_{topo}$ for an annular CSS of $N$ subsystems (eq.\eqref{eq:allcorr}). This generalizes the previous result for the 3-party irreducible correlation \cite{Liu_Zhou_2016,Kato_Murao_2016}. These results demonstrate the presence of $N$-party quantum correlations among the subsystems of an annular CSS, and confirms the existence of closed annular structures of all sizes within a topologically ordered ground state. We conclude with a discussion of our results in Section \ref{section:discussion}. Detailed derivations of several key results are presented in the appendices.

\section{Multipartite information: definition}
\label{section:IN_definition}

\noindent Topologically ordered systems contain nonlocal entanglement and correlations. Thus, identifying nonlocal operators and measures of entanglement is important in their classification. Importantly, in the case of zero correlation length,  the entanglement entropy ($S_{\mathcal{R}}$) of a region $\mathcal{R}$ of a topologically ordered ground state depends on the number of disconnected components ($j$) of the boundaries~\cite{Hamma_Zanardi_2005} ($\partial \mathcal{R}$) of $\mathcal{R}$ and the quantum dimension ($\mathcal{D}$) of the Hilbert space~\cite{Levin_Wen_2006}
\begin{eqnarray}
S_{\mathcal{R}}=-j\log \mathcal{D} - n \displaystyle\sum_{k=0}^{N} \frac{d_k^2}{\mathcal{D}} \log  \bigg(\frac{d_k^2}{\mathcal{D}}\bigg)~,
\label{eq:wenEE}
\end{eqnarray}
and where $n$ is the number of states lying on $\partial \mathcal{R}$. The quantum dimension $\mathcal{D}$ is a property of the complete system, and does not depend on the choice of the subsystems.

\noindent Further, Refs.\cite{Kitaev_Preskill_2006,Levin_Wen_2006} showed that a purely topological part of entanglement entropy (dubbed as topological entanglement entropy (TEE)) can be detected by measuring the tripartite information $I^3$ for a particular annular arrangement of three subsystems: $I_3(A,B,C)=S_{A}+S_{B}+S_{C}-(S_{AB}+S_{BC}+S_{CA})+S_{ABC}= -2S_{topo}$, where $S_{topo}=\log \mathcal{D}$. Using the above formula eq.\eqref{eq:wenEE}, one can easily verify an essential feature of $I^3$: it is defined so as to be independent of the geometry of the arrangement of subsystems; instead, it depends only on the quantum dimension($\mathcal{D}$) of the topologically ordered system. We will now extend this result to show that an appropriately defined $N$-partite information measure can capture the same topological entanglement entropy (TEE) by a careful arrangement of $N$ subsystems. Further, we confine our interest to the case of 2-spatial dimensional topologically ordered systems in this work.

\noindent We first clarify some important mathematical notations and conventions. Our goal is to define the $N-$partite information $I^N_{\{\mathcal{A}_N\}}$ for a collection of subsystems (CSS, with $N$ subsystems) with a unique arrangement specified by $\{\mathcal{A}_N\}\equiv \{A_1,A_2,~.~.~,A_N\}$. Some examples of CSS considered by us are given in Fig.\eqref{fig:island_subsystem_main}, with individual subsystems labelled by $A_1, A_2,~.~.~,A_N$. If there is no overlap between two subsystems $A_i$ and $A_j$, then their intersection vanishes: $A_i\cap A_j=\emptyset$. We define the power set of the CSS $\{\mathcal{A}_N\}$ as $\mathcal{P}(\{\mathcal{A}_{N}\})$, and the collection of all subsets of $\mathcal{P}(\{\mathcal{A}_{N}\})$ with $m$ subsystems in it as $\mathcal{B}_m(\{\mathcal{A}_N\})\equiv \{ Q~| ~Q\subset \mathcal{P}(\{\mathcal{A}_{N}\}), |Q|=m \}$. We also define the union and intersection of all the subsystems present in $Q$ as ${V}_{\cup}({Q})\equiv \bigcup_{A\in Q} A$ and ${V}_{\cap}({Q})\equiv \bigcap_{A\in Q} A$ respectively. Finally, the von Neumann entanglement entropy of the subsystem $A$ (with length $L_{A}$) for a topologically ordered ground state is given by $S_{A}=\alpha L_A-\gamma_A$, and $\gamma_A$ represents the topological terms in $S_{A}$. 

\noindent Then, the $N-$partite information is defined as 
\begin{eqnarray}
I^{N}_{\{\mathcal{A}_N\}} &=& \bigg[\displaystyle\sum_{m=1}^{N} (-1)^{m-1} \displaystyle\sum_{Q \in \mathcal{B}_m(\{\mathcal{A}\})} S_{V_{\cup}({Q})} \bigg]- S_{V_{\cap}(\mathcal{A}_N)}~.~~~~
\label{eq:I_N_definition}
\end{eqnarray}
In this work, we will focus on the cases where there are no overlaps among the $N$ subsystems within the CSS, $V_{\cap}(\{\mathcal{A}_N\})=\emptyset$. In this way, we exclude any nontrivial contributions to the TEE arising from such an overlap in all CSS that we study: $\gamma_{V_{\cap}(\mathcal{A}_N)}=0$. Further, we also assume that there is no overlap among $m$ number of subsystems in the CSS for $N\geq m>2$. Thus, for an example of $m=3$, $A_i \cap A_j \cap A_k =\emptyset,~\forall i\neq j\neq k$. One can easily check that for $N=3$, eq.\eqref{eq:I_N_definition} then becomes the tripartite information $I^{3}$ given above for an annular structure of a CSS of three subsystems (i.e., with a hole in the center) \cite{Kitaev_Preskill_2006,Levin_Wen_2006}. 
\par\noindent We will now demonstrate that the multipartite information measure $I^N_{\{\mathcal{A}_N\}}$ is chosen such that all geometric content within it vanish identically. For this, we first define the geometric area of a subsystem $A$ as $R({A})$. Then, the geometry dependence ($\mathcal{R}_{\{\mathcal{A}_N\}}$) of $I^N_{\{\mathcal{A}_N\}}$ is computed as follows
\begin{eqnarray}
&& \mathcal{R}_{\{\mathcal{A}_N\}} = \displaystyle\sum_{m=1}^{N} (-1)^{m-1} \displaystyle\sum_{Q \in \mathcal{B}_m(\{\mathcal{A}\})} R({V_{\cup}({Q})}) ~,~~\nonumber\\
&& ~~~~= \displaystyle\sum_{m=1}^{N} \displaystyle\sum_{i=1}^{N} (-1)^{m-1} {N-1 \choose m-1}\times R(A_i) \nonumber\\
&& ~~~~= \bigg[ \displaystyle\sum_{i=1}^{N} R(A_i) \bigg] \times \bigg[\displaystyle\sum_{m=1}^{N} (-1)^{m-1} {N-1 \choose m-1} \bigg] =0~.~~
\label{eq:geometryind}
\end{eqnarray}
\par\noindent Thus, we find that $I^N_{\{\mathcal{A}_N\}}$ is indeed independent of the geometry of its constituents, i.e., the last term in eq.\eqref{eq:wenEE} (related to the number of the states $n$ in the subsystem perimeter) cancel one other within the measure $I^N$ (eq.\eqref{eq:I_N_definition}). Thus, for topologically ordered systems, the multipartite information measure $I^N$ will depend only on $\log \mathcal{D}$, and with a prefactor that depends on the choice of subsystems
\begin{eqnarray}
I^{N}_{\{\mathcal{A}_N\}} &=& -\mathcal{C}^N_{\{\mathcal{A}_N\}} \log \mathcal{D}~,~~\\
\mathcal{C}^N_{\{\mathcal{A}_N\}}&=& \displaystyle\sum_{m=1}^{N} (-1)^{m-1} \displaystyle\sum_{Q \in \mathcal{B}_m(\{\mathcal{A}\})} {\mathcal{J}}_{V_{\cup}({Q})} ~,~~~~
\label{eq:N-partite-connectivity}
\end{eqnarray}
\noindent where $\mathcal{J}_{\mathcal{B}}$ is the number of disconnected/disjoint boundaries of the subsystem $\mathcal{B}$ or the number of disconnected components of $\partial\mathcal{B}$ (the boundary of $\mathcal{B}$). 
\par\noindent 
Indeed, we will see in the following section that the quantity $\mathcal{C}^N_{\{\mathcal{A}_N\}}$ quantifies the nontrivial topology of the CSS. We demonstrate both trivial as well as nontrivial choices for the topology of a CSS, and describe the transformations that leave $I^N_{\mathcal{A}_N}$ invariant. This will generalise the results of Ref.\cite{Kitaev_Preskill_2006,Levin_Wen_2006} on how to detect the TEE of a CSS of $N$ subsystems via a $N-$partite information measure, and offer insights into the nature of the many-particle entanglement encoded in such systems.

\section{Multipartite information: computation}
\label{section:I_N_computation}

\par\noindent We confine ourselves to the study of CSS $\{\mathcal{A}_N\}$ that are placed on a 2D planar manifold. Thus, the individual subsystems are 2-dimensional, and their boundaries are 1-dimensional curves. As mentioned earlier, we also assume that there exist no overlaps among different subsystems other than nearest neighbours. Starting with the simplest case of an annulus, we now compute the $I^{N}$ for several different arrangements of the CSS in order to understand the role played by subsystem topology.
\begin{figure}[!h]
\includegraphics[scale=0.28]{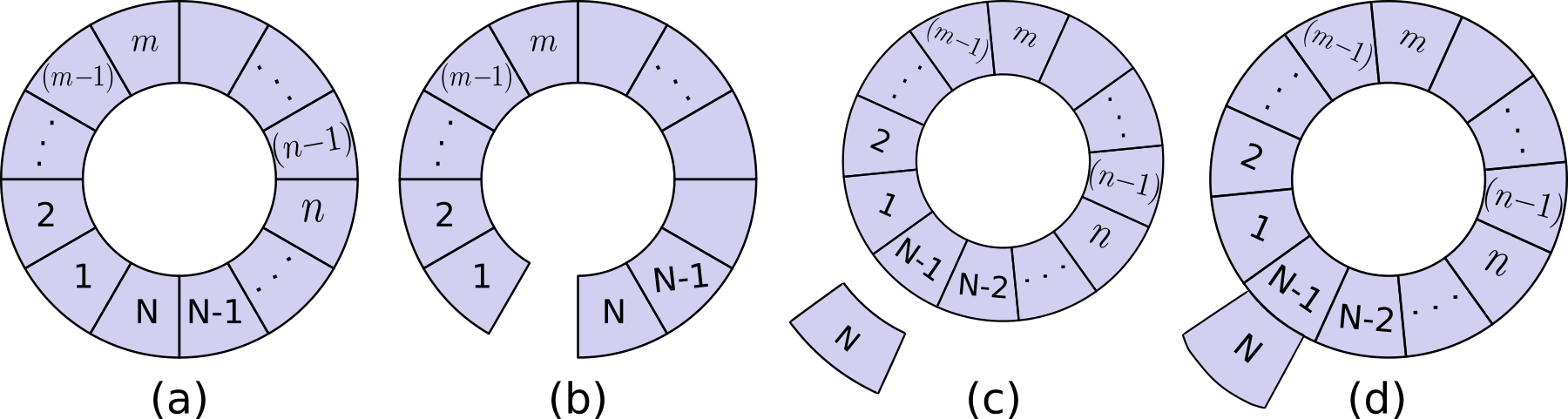}
\caption{(a) CSS with $N$ members placed in an annular structure ($\{\mathcal{A}^{(a)}_N\}\equiv\{A_1,~.~.~,A_N\}$), and where each individual subsystem is simply connected. (b) An open CSS ($\{\mathcal{A}^{(1b)}\}$) formed out of $N$ number of simply connected subsystems, such that their union does not form a closed annular structure. 
(c) CSS ($\{\mathcal{A}^{(1c)}\}$) composed of $N$ subsystems, where $N-1$ form a closed CSS ($\{\mathcal{A}^{(a)}_{N-1}\}$) and the last is an isolated island $A_N$. (d) CSS ($\{\mathcal{A}^{(1d)}\}$) formed from $N+1$ subsystems created by joining the (appendage) subsystem $A_{N-1}$ with the closed annular CSS $\{\mathcal{A}^{(a)}_{N-1}\}$.}
\label{fig:island_subsystem_main}
\end{figure}

\subsection{Simple annular closed and open structures}
\label{bm:simple_annular}
\par\noindent 
We first study the simplest CSS shown in the Fig.\ref{fig:island_subsystem_main}(a), where the subsystems $\{\mathcal{A}^{(a)}_N\}=\{A_1,A_2,~.~.~,A_N\}$ are arranged in an annulus. Each subsystem $A_i\in \{\mathcal{A}^{(a)}_N\}$ has a single disconnected boundary, $\mathcal{J}_{A_i}=1$~. One can also easily see that for such a CSS,  $V_{\cap}(\{\mathcal{A}^{(a)}_N\})=\emptyset$. Then, the quantity
\vspace{-0.30cm}
\begin{eqnarray}
\mathcal{C}^N_{\{\mathcal{A}^{(a)}_N\}}&=& \displaystyle\sum_{m=1}^{N}  (-1)^{m-1}~ \Gamma^N_{m}~, \nonumber\\
~~\Gamma^N_{m}&=& \displaystyle\sum_{Q \in \mathcal{B}_m(\{\mathcal{A}_N\})} {\mathcal{J}}_{V_{\cup}({Q})} ~,~~~~
\end{eqnarray}
\noindent where $\Gamma^N_{m}$ represents the total number of disconnected boundaries coming from all possible choices of $m$ subsystems (out of a total $N$ subsystems). One of our main results involves computing the count $\Gamma^N_m$. As shown in Appendix \eqref{ap:cycle_graph}, for an annular arrangement of subsystems $\{\mathcal{A}^{(a)}_N\}$, $\bigg[\displaystyle\sum_{m=1}^{N-1}  (-1)^{m-1}\displaystyle\sum_{Q \in \mathcal{B}_m(\{\mathcal{A}^{(a)}_N\})} {\mathcal{J}}_{V_{\cup}({Q})}\bigg]=0$. Thus, we obtain 
\begin{eqnarray}
\mathcal{C}_{\{\mathcal{A}^{(a)}_N\}}^{N} &=& (-1)^{N-1} \Gamma_{N}^{N} = (-1)^{N-1}\mathcal{J}_{\cup_i A_i}= 2\times (-1)^{N-1} ~,~
\end{eqnarray}
as the number of disconnected boundaries of the entire CSS as a whole (i.e., $\bigcup_i A_i$ in the CSS $\{\mathcal{A}^{(a)}_N\}$) is given by $\Gamma^N_N=\mathcal{J}_{\cup_i A_i}=2$ (see Fig.\ref{fig:island_subsystem_main}(a)). We have also shown in Appendix (\ref{ap:simpleannular_euler}) that the prefactor $2$ corresponds to the Euler characteristics $\chi$ of the CSS embedded on the underlying planar manifold. Thus, the $N-$partite information simplifies to 
\begin{eqnarray}
I^N_{\{\mathcal{A}^{(a)}_{N}\}}&=&-\mathcal{C}^N_{\{\mathcal{A}^{(a)}_{N}\}} \log \mathcal{D}=(-1)^{N}2\log\mathcal{D}~, \nonumber\\
&=& (-1)^N \chi S_{topo}~, ~~~~
\label{eq:IN_annular}
\end{eqnarray}
\noindent Thus, we see that for a simple annular arrangement of subsystems, the amplitude of the $N-$partite information has the same value $|I^N_{\{\mathcal{A}^{(a)}_{N}\}}|=2\log\mathcal{D}$ for all $N$. For the case of $N=3$, this reduces to the well known result for the tripartite information \cite{Kitaev_Preskill_2006,Levin_Wen_2006}. Our generalization highlights a property likely special to a topologically ordered system: any $N-$partite information ($N\geq 3$) is able to capture the TEE of $S_{topo}$.

\par\noindent On the other hand, if the structure of the CSS is open (see Fig.\ref{fig:island_subsystem_main}(b)), the $N-$partite entanglement measure will vanish even for a topologically ordered ground state: $I^N_{\{\mathcal{A}^{(1b)}_{N}\}}=0$. This is shown in Appendix \eqref{ap:path_graph}. Again, this shows the crucial role of the subsystem topology of the CSS in capturing $S_{topo}$.
\subsection{Isolated subsystems and appendages}
\label{section:isolated_island}

\par \noindent 
We now turn our attention to the case of a CSS comprised of a disjoint union of an annular arangement of $N-1$ subsystems and an isolated subsystem labelled by $N$ as shown in the Fig.\ref{fig:island_subsystem_main}(c). We use eq.\eqref{eq:N-partite-connectivity} to calculate $\mathcal{C}^N_{\{\mathcal{A}^{(1b)}_N\}}$, where ${\mathcal{J}}_{V_{\cap} (\{\mathcal{A}^{(1b)}_N\})}=0$. Recall that
\begin{eqnarray}
\mathcal{C}^N_{\{\mathcal{A}^{(1b)}_N\}}&=& \displaystyle\sum_{m=1}^{N} (-1)^{m-1} \displaystyle\sum_{Q \in \mathcal{B}_m(\{\mathcal{A}^{(1b)}_{N}\})} {\mathcal{J}}_{V_{\cup}({Q})}  ~.~
\label{eq:css_open}
\end{eqnarray}
As the CSS $\{\mathcal{A}^{(1b)}_{N}\}$ is a disjoint union of two smaller CSS, $\{\mathcal{A}^{(1b)}_{N}\}=\{\mathcal{A}^{(1b),N}_{N-1}\}\cup \{A_N\}$. Thus, upon expansion of eq.\eqref{eq:css_open}, we can rewrite $\mathcal{C}^N_{\{\mathcal{A}^{(1b)}_N\}}$ as
\begin{eqnarray}
&&\mathcal{C}^N_{\{\mathcal{A}^{(1b)}_N\}} = \mathcal{C}^{N-1}_{\{\mathcal{A}^{(1b)}_N\}-\{A_N\}}+\mathcal{J}_{A_N} -\displaystyle\sum_{i=1}^{N-1} \mathcal{J}_{A_N\cup A_i}\nonumber\\
&&~~~+ \displaystyle\sum_{i<j=1}^{N-1} \mathcal{J}_{A_N\cup A_i\cup A_j} ~.~.~ +(-1)^{N-1}\mathcal{J}_{A_N\cup A_1\cup ~.~ \cup A_{N-1}}\nonumber\\
&&= \mathcal{C}^{N-1}_{\{\mathcal{A}^{(1b)}_N\}-\{A_N\}}-\mathcal{C}^{N-1}_{\{\mathcal{A}^{(1b)}_N\}-\{A_N\}}+\mathcal{J}_{A_N} \Upsilon_{N-1} \nonumber\\ 
&&= 0~,
\label{inapp}
\end{eqnarray}
where $\Upsilon_{N-1}=\sum_{i=0}^{N-1} (-1)^i~~ {{N-1}\choose i} =0$. The detailed derivation of the relation eq.\eqref{inapp} is shown in the Appendix (\ref{ap:isolated_Structure}). Thus, one can see that for a CSS of $N$ subsystems that can be decomposed into disjoint sub-structures, a vanishing global connectivity measure $C^{N}$ gives rise to a vanishing $N$-partite entanglement measure $I^{N}$.

\par\noindent We have also considered a CSS of $N$ number of subsystems in which an appendage has been added to the simple annular structure (see Fig.\ref{fig:island_subsystem_main}(d)). As shown in Appendix \eqref{ap:Appendage}, the multipartite information vanishes for this CSS as well: $I^N_{\{\mathcal{A}^{(1d)}_{N}\}}=0$. In this way, we have found that the closed annular nature of the CSS is essential for a nontrivial value of $I^{N}$. We will now analyse deformations of the CSS that keep the result $I^N_{\{\mathcal{A}\}}=(-1)^{N}\chi\log\mathcal{D}~$~  invariant.

\subsection{Individual subsystem boundary with multiple disconnected components}
\label{section:multiple_disconnected_components}

\noindent We are interested in those cases of CSS with $N$ subsystems where individual subsystems have either holes and/or self-handles (as shown in Fig.\ref{fig:island_multiplyConnected}(a)). Unlike the simple annular case shown in Fig.\ref{fig:island_subsystem_main}(a), the number of disconnected boundaries of an individual subsystem can in such cases be an integer value higher than 1. We represent such CSS by $\{\mathcal{A}^{(2a)}_{N}\}$.
\begin{figure}[!h]
\includegraphics[scale=0.36]{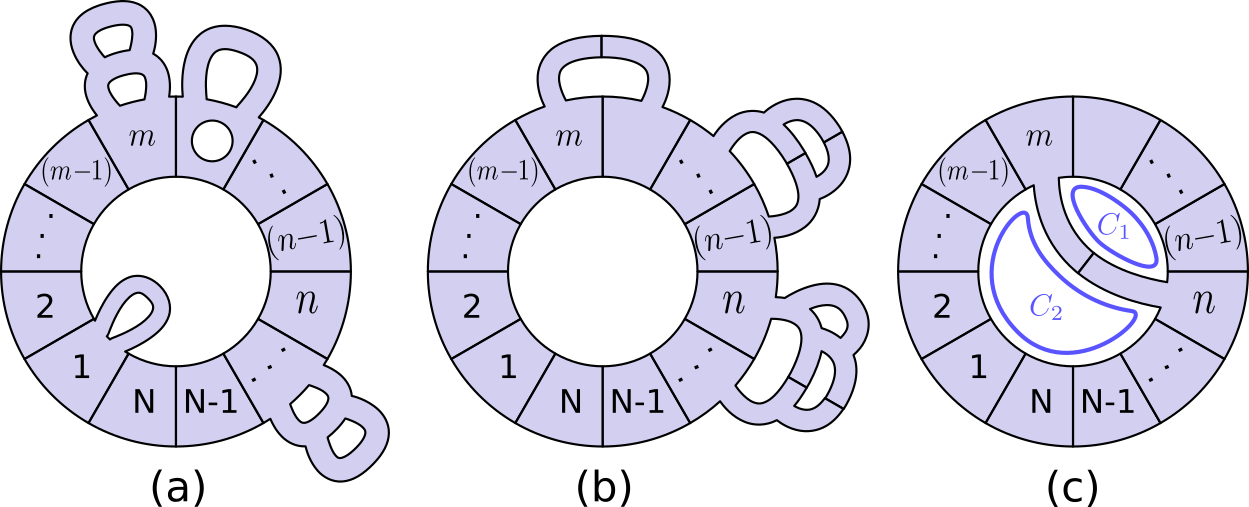}
\caption{(a) CSS $\{\mathcal{A}^{(2a)}_{N}\}$ represents the case where individual subsystems have holes and multiple handles, but there is no inter subsystem handles. (b) CSS representing the case where there are handles among nearest neighbour subsystems, but no connection among further neighbour subsystems. (c) CSS with a connection between two subsystems ($A_m$ and $A_n$) that lie beyond nearest neighbour ($n-m>1$).}
\label{fig:island_multiplyConnected}
\end{figure}
\par\noindent
We now calculate the quantity $\mathcal{C}^N_{\{\mathcal{A}^{(2a)}_{N}\}}$ as a deviation from $\mathcal{C}^N_{\{\mathcal{A}^{(a)}_{N}\}}$, due to the increase in the number of disconnected boundaries of the individual subsystems. Note that here, $\mathcal{J}_{A_i}=1+\mu_i$, where $\mu_i$ is the change in the number of disconnected boundaries with respect to simple annular case (see Fig.\eqref{fig:island_subsystem_main}(a) and the discussion in Section \ref{bm:simple_annular}). Thus, we obtain
\begin{eqnarray}
\mathcal{C}^N_{\{\mathcal{A}^{(2a)}_N\}}&=& \displaystyle\sum_{m=1}^{N} (-1)^{m-1} \displaystyle\sum_{Q \in \mathcal{B}_m(\{\mathcal{A}^{(2a)}_{N}\})} {\mathcal{J}}_{V_{\cup}({Q})}  ~,~\nonumber\\
&=& \mathcal{C}^N_{\{\mathcal{A}^{(a)}_N\}} + \displaystyle\sum_{m=1}^{N} (-1)^{m-1}~ \Xi^N_m ~,~~\nonumber\\
 \Xi^N_{m} &=& \displaystyle\sum_{Q \in \mathcal{B}_m(\{\mathcal{A}^{(2a)}_{N}\})} \displaystyle\sum_{A_i\in Q} \mu_i= {N-1 \choose m-1} \displaystyle\sum_{A_i\in \{\mathcal{A}^{(2a)}_{N}\}} \mu_i~.~~~~~~~\nonumber
\end{eqnarray}

One can easily simplify this result using the relation 
\begin{eqnarray}
\displaystyle\sum_{m=1}^{N} (-1)^{m-1} {N-1 \choose m-1} =0~,
\end{eqnarray}
thereby leading to $\mathcal{C}^N_{\{\mathcal{A}^{(2a)}_N\}}= \mathcal{C}^N_{\{\mathcal{A}^{(a)}_N\}}$. Then, using eq.\eqref{eq:IN_annular}, we again obtain 

\begin{eqnarray}
I^N_{\{\mathcal{A}^{(2a)}_N\}}= I^N_{\{\mathcal{A}^{(a)}_N\}}=(-1)^N \chi S_{topo}~.~~~~~
\end{eqnarray}
Thus, we find that the addition of self-handles and holes in the individual subsystems does not affect our earlier result for the $N-$partite information measure $I^{N}$. This confirms the robustness of our multipartite information measure in topologically ordered phases against such a deformation.

\subsection{Adding nearest-neighbour handles}
\label{subsec:nnhandles}
\noindent  As shown in Fig.\ref{fig:island_multiplyConnected}(b), we now deform the simple annular structure (Fig.\ref{fig:island_subsystem_main}(a)) by adding any number of inter subsystem handles between nearest neighbour subsystems. The deformed CSS is denoted by $\{\mathcal{A}^{(2b)}_{N}\}$. 
 
In order to compute the $N-$partite information $I^N_{\{\mathcal{A}^{(2b)}_{N}\}}$, we first calculate 
$\mathcal{C}^N_{\{\mathcal{A}^{(2b)}_{N}\}}$. Recall that $\mathcal{J}_{A_i\cup A_{i+1\textrm{mod}N}}=1$ for the simple annular case. Due to addition of extra $\nu_i$ number of handles between the nearest neighbour subsystems $A_i$ and $A_{i+1\textrm{mod}N}$, we have increased the number of disconnected boundary to $\mathcal{J}_{A_i\cup A_{i+1\textrm{mod}N}}=1+\nu_i$. Thus ,
\begin{eqnarray}
\mathcal{C}^N_{\{\mathcal{A}^{(2b)}_{N}\}} &=& \mathcal{C}^N_{\{\mathcal{A}^{(a)}_{N}\}} + \displaystyle\sum_{m=2}^{N} (-1)^{m-1} \displaystyle\sum_{i=1}^{N} \nu_i {N-2 \choose m-2} ~ ~~\nonumber\\
&=& \mathcal{C}^N_{\{\mathcal{A}^{(a)}_{N}\}} + \displaystyle\sum_{i=1}^{N} \nu_i\displaystyle\sum_{m=2}^{N} (-1)^{m-1}  {N-2 \choose m-2} ~ ~~\nonumber\\
&=& \mathcal{C}^N_{\{\mathcal{A}^{(a)}_{N}\}}~.
\end{eqnarray}
Thus the $N$-partite information measure is also invariant under this deformation  
\begin{eqnarray}
I^N_{\{\mathcal{A}^{(2b)}_{N}\}} &=& I^N_{\{\mathcal{A}^{(a)}_{N}\}}= (-1)^N \chi S_{topo}~.~~~~~~
\end{eqnarray}
\subsection{Addition of further-neighbour handles} 
\label{section:morethan_nearest_neighbour}

\par\noindent Having analysed deformations of the CSS that leave the multipartite information $I^{N}$ invariant, we now turn to a deformation that trivializes it. For this, we add a single handle between two subsystems $A_m$ and $A_{n}$ with atleast one subsystem lying in between them (such that $A_m$ and $A_n$ are not nearest neighbours, see Fig.\ref{fig:island_multiplyConnected}(c)). We choose the subsystem label such that $n< m \leq N$, where $(m-n)>1$. It is easily seen that upon adding such a handle, we create two closed loops $C_1$ and $C_2$ formed out of $p$ and $q$ number of subsystems respectively, where $p=(m-n+1)$ and $q=(N-m+n+1)$ with $q>p$ and $p+q-2=N$. Thus, we can now create simple annular CSS (of the kind seen in Fig.\eqref{fig:island_subsystem_main}(a)) from the closed loops $C_1$ and $C_2$, and denoted by $\{\mathcal{A}^p\}$ and $\{\mathcal{A}^q\}$ respectively. We now compute 
\begin{eqnarray}
\mathcal{C}^N_{\{\mathcal{A}^{(2a)}_N\}}&=& \displaystyle\sum_{m=1}^{N} (-1)^{m-1}\bigg[\mu_m + \displaystyle\sum_{Q \in \mathcal{B}_m(\{\mathcal{A}^{(a)}_{N}\})} {\mathcal{J}}_{V_{\cup}({Q})}  \bigg] ~,~
\end{eqnarray}
where the modification terms $\mu_m$ (apparent upon the introduction of the handle) are given by
\begin{eqnarray}
\mu_1 &=& 0~,~\mu_N=+1~,~ 
\mu_{m} = -{{N-2}\choose{j-2}} ~ ,~\forall~  2 \le m<p~, \nonumber\\
\mu_{m} &=& 2 {{N-p}\choose{j-p}} -{{N-2}\choose{j-2}}~,~ \forall~ p\le m <q~, \nonumber\\
\mu_{m} &=& 2{{N-p}\choose{j-p}}+2{{N-q}\choose{j-q}} -{{N-2}\choose{j-2}}~, \nonumber\\
&&~~~~~~~~~~~~~~~~~~~~~~~~~~~~~~~~~~~~~~~ \forall~ q \le m < N ~.~~~~
\end{eqnarray}
\noindent Thus, we obtain $\mathcal{C}^N_{\{\mathcal{A}^{(2a)}_N\}}=\mathcal{C}^N_{\{\mathcal{A}^{(a)}_N\}}+ \displaystyle\sum_{m=1}^{N} (-1)^{m-1} \mu_m$. Further simplification gives 
\begin{eqnarray}
\mathcal{C}^N_{\{\mathcal{A}^{(2a)}_N\}}&=&\mathcal{C}^N_{\{\mathcal{A}^{(a)}_N\}}+ 2(-1)^N\nonumber\\
&=& 2(-1)^{N-1}+2(-1)^N=0~.~~~~~
\end{eqnarray}
Additionally, in Appendix \ref{ap:holes_in_css}, we also show that the vanishing of $\mathcal{C}^N_{\{\mathcal{A}^{(2a)}_N\}}$ in this case arises from the fact that $\mathcal{C}^N_{\{\mathcal{A}^{(2a)}_N\}} = (-1)^{N}(\chi - 2)$, where $\chi \equiv 2$ is the Euler characteristic of the CSS with $N$ subsystems embedded on the planar manifold.
In turn, the vanishing $\mathcal{C}^N$ of leads to the vanishing of the $N-$partite information for the entire CSS 
\begin{equation}
I^N_{\{\mathcal{A}^{(2c)}_N\}}= -\mathcal{C}^N_{\{\mathcal{A}^{(2a)}_N\}}\log\mathcal{D} = 0~.
\label{eq:allzero}
\end{equation}
Instead, the $p$-partite and $q-$partite information for the CSS $\{\mathcal{A}^p\}$ and $\{\mathcal{A}^q\}$ (i.e., the two smaller loops $C_1$ and $C_2$) are found to be non-zero as long as $p,q\geq 3$:
\begin{equation}
I^p_{\{\mathcal{A}^p\}}=(-1)^p\chi\log\mathcal{D}~,~I^q_{\{\mathcal{A}^q\}}=(-1)^q\chi\log\mathcal{D}~.
\label{eq:newinfo}
\end{equation}
\par\noindent 
\begin{figure}[!h]
\centering
\includegraphics[scale=0.196]{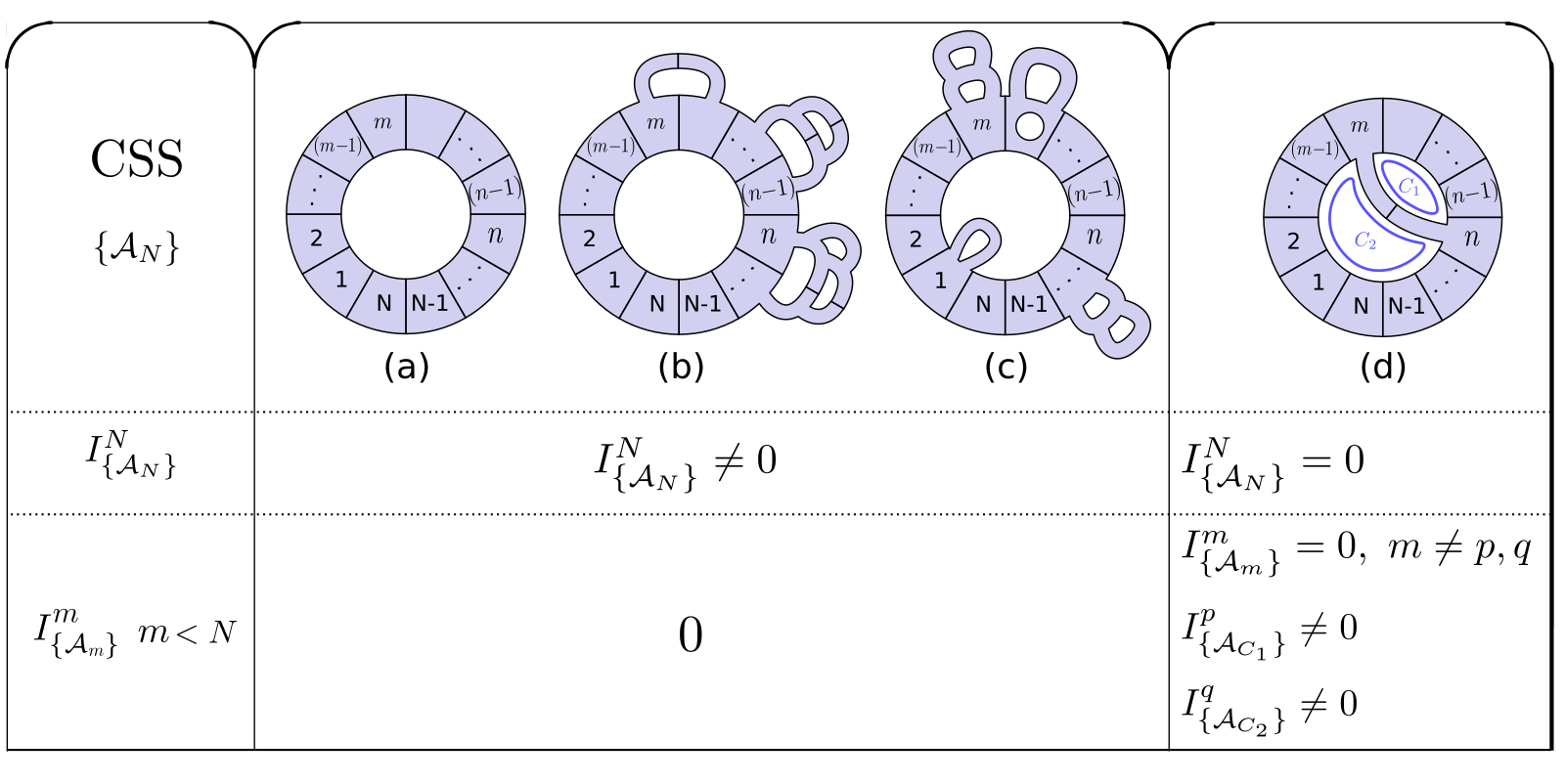}
\caption{Summary of various results for the $N$-partite information $I^{N}$ corresponding to different CSS topologies placed on a planar manifold presented in Section \ref{section:I_N_computation}. Please refer to the text for details. }
\label{fig:summary}
\end{figure}

\par\noindent 
Finally, we summarise the results of this section in Fig.\ref{fig:summary}. The simple annular CSS ((a) in Fig.\ref{fig:summary}) with $N$ number of subsystems possesses a non-zero multipartite information $I^{N}$, dependent on the Euler characteristic ($\chi$) of the CSS embedded on the underlying manifold and the quantum dimension of the ground state manifold ($\mathcal{D}$). Deformations of this simple annular structure that involve the addition of intra subsystem handles or holes ((b) in Fig.\ref{fig:summary}) and the addition of nearest neighbour handles ((c) in Fig.\ref{fig:summary}) leave the multipartite information measure $I^{N}$ invariant. On the other hand, the addition of the further-neighbour handles ((d) in Fig.\ref{fig:summary}) trivializes the $I^N$ measure. However, the addition of the further-neighbour handles creates two closed loops $C_1$ and $C_2$ ((e) and (f) in Fig.\ref{fig:summary}) with a lesser number of sub-sysems ($p$ and $q$ respectively, corresponding to the CSS $\{A_p\}$ and $\{\mathcal{A}_q\}$). The smaller loops $C_1$ and $C_2$ again possess a non-zero multipartite information (as long as $p,q\geq 3$). In this way, we find that it is always possible to identify a simple annular structure with an $I^{N}$ that can detect the topological entanglement entropy.
\section{Measurement Constraints}
\label{section:measurement_constraint}

\par\noindent 
Having ascertained the importance of subsystem topology in attaining a nontrivial multipartite information measure $I^{N}$, we now turn our attention to measurements for a more general case of a CSS that has more than one hole in it, i.e., composed of more than one annular structure. 
\begin{figure}[!h]
	\includegraphics[scale=0.6]{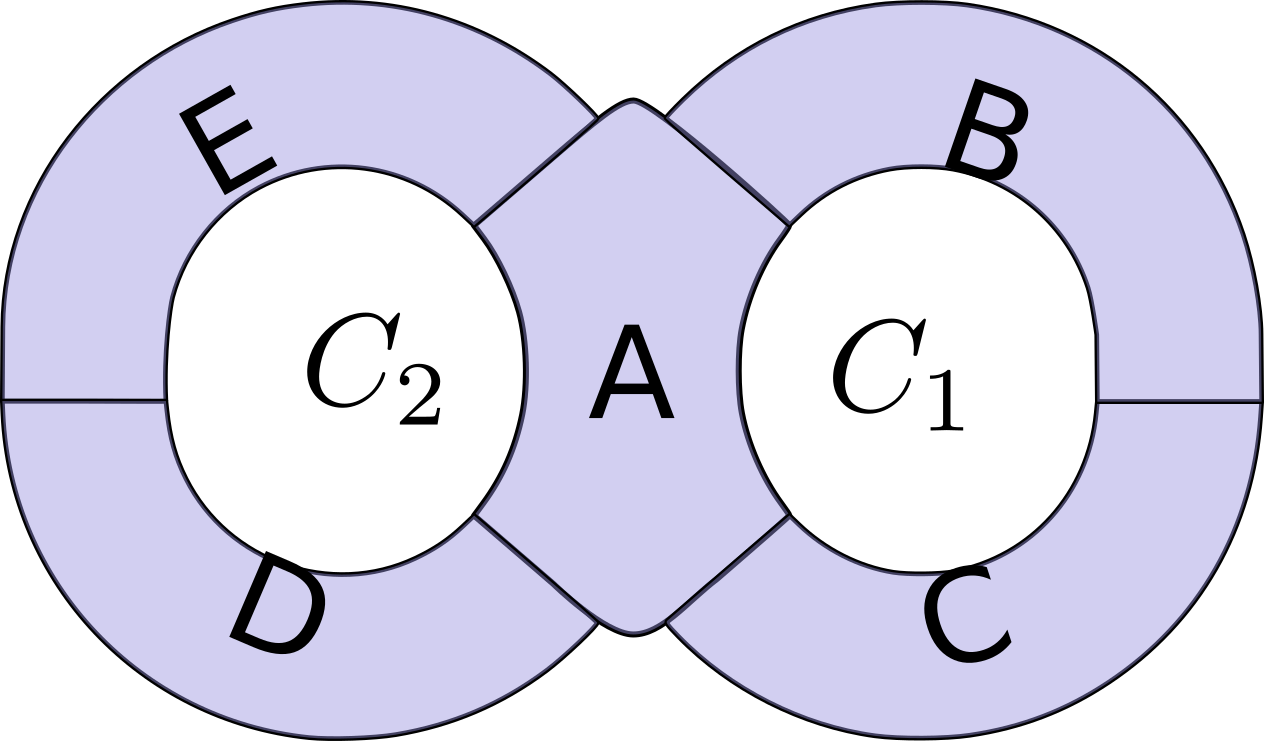}
	\caption{CSS with 5 subsystems ($\{\mathcal{A}^{(2)}_{5}\}=\{A,B,C,D,E\}$) and embedded on planar manifold. The two holes in the CSS are denoted by $C_1$ and $C_2$.}
	\label{fig:genus_2_1}
\end{figure}
As an example, we start with a minimally complex CSS $\{\mathcal{A}^{(2)}_{5}\}$ that has 2 holes in it (see Fig.\eqref{fig:genus_2_1}). Two closed annular structure is $C_1$ and $C_2$ associated with the smaller CSS $\{\mathcal{A}^{C_1}_{3}\}=\{A,B,C\}$ and $\{\mathcal{A}^{C_2}_{3}\}=\{A,D,E\}$ respectively. 
Using eq.\eqref{eq:IN_annular} for the CSS $\mathcal{A}^{C_1}_{3}$ and $\mathcal{A}^{C_2}_{3}$ shown in Fig.\eqref{fig:genus_2_1}, it is easily seen that 
\begin{eqnarray}
I^3_{\{\mathcal{A}^{C_1}_{3}\}}=I^3_{\{\mathcal{A}^{C_2}_{3}\}}=(-1)^3 \chi\log \mathcal{D}~.~
\end{eqnarray}
\par\noindent Now, our goal is to calculate $I^N_{\{\mathcal{A}^{(2)}_5\}}$ for the CSS $\{A^{(2)}_5\}$ shown in Fig.\eqref{fig:genus_2_1}. For this, we first note that the $N-$partite information for any general CSS $\{A\}$, as shown in the eq.\eqref{eq:I_N_definition}, can be re-written in terms of various lower order multipartite information measures as follows
\begin{eqnarray}
I^{N}_{\{A_N\}}&=& \displaystyle\sum_{\mu=1}^{N-2}(-1)^{\mu-1} \displaystyle\sum_{\substack{R\in \mathcal{B}_{N-\mu}(\{A\})\\ |R|=N-\mu}} I^{|R| }_{R}  \nonumber\\
&&~~~~~~~~~~~~~~ + (-1)^{N} \bigg(  \displaystyle\sum_{i} S_{A_i} - S_{\cup A_i}  \bigg)~.~
\label{eq:I_N_recursion}
\end{eqnarray}
The detailed derivation for eq.\eqref{eq:I_N_recursion} has been shown in Appendix \ref{app:recursionIN}.
Using eq.\eqref{eq:I_N_recursion}, we obtain
\begin{eqnarray}
&&I^5_{\{A^{(2)}_5\}}= \displaystyle\sum_{\{a\}\in \mathcal{B}_{4}(\{A^{(2)}_5\})} I^4_{\{a\}}-\displaystyle\sum_{\{b\}\in \mathcal{B}_{3}(\{A^{(2)}_5\})} I^3_{\{b\}}\nonumber\\
&&+\displaystyle\sum_{\{c\}\in \mathcal{B}_{2}(\{A^{(2)}_5\})} I^2_{\{c\}}  +(-1)^{N} \bigg(  \displaystyle\sum_{i} S_{A_i} - S_{\cup A_i}  \bigg)~.~~~~~~
\end{eqnarray}
\noindent One can easily check that all the $I^4_{\{..\}}$ terms vanish identically, as there is no possible closed annular region in Fig.\eqref{fig:genus_2_1} with $4$ subsystems and where each subsystem has two unique nearest neighbours. Some of the $I^3_{\{..\}}$ terms also vanish for the same reason, leaving only two non-vanishing $I^{3}$ terms: $I^3_{\{A,B,C\}}\neq 0 \neq I^3_{\{A,D,E\}}$. Further, only some of the ${5 \choose 2}$ number of $I^2_{\{..\}}$ terms vanish, as $\mathcal{C}^2_{A_i,A_j}=0$ if $\mathcal{J}_{{A_i}\cup A_j} \neq 1$, and $\mathcal{C}^2_{A_i,A_j}=1$ otherwise. Also, from its definition, we know that $\mathcal{J}_{A}=1$ for a subsystem $A$ with one disconnected boundary. Finally, we note that $\mathcal{J}_{\cup_i A_i}=3$, as the number of disconnected boundaries of the entire CSS ($\cup_i A_i$) in the Fig.\eqref{fig:genus_2_1} is $3$. Using these properties, we find
\begin{eqnarray}
&&I^5_{\{A^{(2)}_5\}}=-(I^3_{\{A,B,C\}}+I^3_{\{A,D,E\}} )\nonumber\\
&&+ (I^2_{\{A,B\}}+I^2_{\{B,C\}}+I^2_{\{C,A\}}+I^2_{\{A,D\}}+I^2_{\{D,E\}}+I^2_{\{E,A\}}) \nonumber\\
&& ~~~~~~~~~~~~~~~~ +(-1)^{5} \bigg(  \displaystyle\sum_{i} S_{A_i} - S_{\cup A_i}  \bigg)=0~.~~~~~~
\end{eqnarray}
This can be re-written as
\begin{eqnarray}
&&I^3_{\{A,B,C\}}+I^3_{\{A,D,E\}}= (-1)^{5}\bigg[\displaystyle\sum_{i} S_{A_i} - S_{\cup A_i}\bigg]+\bigg[I^2_{\{A,B\}}\nonumber\\
&&~~~~+I^2_{\{B,C\}}+I^2_{\{C,A\}}+I^2_{\{A,D\}}+I^2_{\{D,E\}}+I^2_{\{E,A\}}\bigg]~, ~~~~~~ \nonumber\\
&&~~~~~~~~= -2\chi\log\mathcal{D}~,~~ \nonumber \\
&&~~~~\Rightarrow |I^3_{\{A,B,C\}}|+|I^3_{\{A,D,E\}}| = 2\chi\log\mathcal{D}~.  
\label{eq:inf_cons_}
\end{eqnarray}

\noindent Following the calculation shown in Appendix \ref{ap:entanglemetn_constraint}, we identify the factor of $2$ in eq.\eqref{eq:inf_cons_} as arising from the two holes in the subsystem configuration. For a general configuration with $n_{h}$ number of holes, we find that the sum of the $n_{h}$ multipartite informations computed around the $n_{h}$ holes is related to $S_{topo}$ as 
\begin{eqnarray}
\displaystyle\sum_{j=1}^{n_h} |I^{\mu_j}_{\{A_j\}}|=\chi n_h \log\mathcal{D} = \chi n_h S_{topo}~,
\label{eq:alltopo}
\end{eqnarray}
where $\mu_{j}$ counts the number of subsystems around the $j$th hole, and $\chi$ is the Euler characteristic of the CSS embedded on the planar manifold ($\chi\equiv 2$).
 
We recall that a similar calculation (see Appendix \ref{ap:holes_in_css}) for the multipartite information $I^N_{\{\mathcal{A}^{(n_h)}_N\}}$ of the complete CSS $\{\mathcal{A}^{(n_h)}_N\}$ with multiple number of holes is observed to vanish
\begin{eqnarray}
I^N_{\{\mathcal{A}^{(n_h)}_N\}}&=& (-1)^{(N-1)} 
(\chi - 2)\log\mathcal{D} = 0~.
\end{eqnarray}
This is a manifestation of the result shown in eq.\eqref{eq:allzero}.

\par\noindent Thus, one can see the multipartite information measurements around different holes of the CSS embedded on the planar manifold (eq.\eqref{eq:alltopo}) are constrained through the interplay of subsystem topology (i.e, the number of holes in the CSS, $n_{h}$) and the Euler characteristic ($\chi$) of the underlying manifold.

\section{TEE and irreducible correlations}
\label{section:TEE_irreducible_correlations}

\par\noindent If a $N-$partite state has multipartite entanglement, then there should also exist signatures of multipartite quantum correlations among the subsystems. Indeed, such correlations can be of any order within the $N-$partite state, ranging from 2-particle to $N$-particle. Following Refs.\cite{Liden_Wootters_2002,Zhou_2008,Kim_2021,Zhou_You_2006,Liu_Zhou_2016,Kato_Murao_2016}, we now seek the connection between the multipartite information $I^{N}$ and the irreducible quantum correlations (defined below) for the simple annular arrangement of the CSS ($\{\mathcal{A}^{(a)}_N\}$; see Fig.\eqref{fig:island_subsystem_main}(a)) of a topologically ordered system. We start with a $N-$partite quantum state $\rho_{\cup_i A_i}$ in the state space $\mathcal{S}(\mathcal{H}^N)$, where $\mathcal{H}^N$ is the Hilbert space and $\{\mathcal{A}^{(a)}_N\}=\{A_1,~.~.~,A_N\}$ is the set of the subsystems forming an annular structure. The $k-$partite irreducible quantum correlation is defined for a $N$-partite state ($k\leq N$), and measures the correlation present purely in the $k-$particle reduced density matrix but not in the $l-$particle reduced density matrix for $l<k$. As we will see below, the $k$-partite irreducible quantum correlations can be defined \cite{Liden_Wootters_2002,Zhou_2008} by using the notion of a maximum entropy state~. 

\par\noindent We first define the set $R_k$
\begin{eqnarray}
R_k &=& \{\sigma~|~\forall a_k\subset \{\mathcal{A}^{(a)}_N\}, |a_k|=k: \sigma_{a_k}=\rho_{a_k}\}~ ~~~~~~~ \nonumber\\
&& \tilde{\rho}^{(k)}_{\{\mathcal{A}^{(a)}_N\}} \equiv  \underset{\sigma\in R_k}{\textrm{argmax}} ~ S(\sigma)~, 
\label{eq:R_k}
\end{eqnarray}
where $|a_k|$ is the cardinality of the set $a_k$ and $\mathscr{S}(\mathcal{H}^N)$ is the state corresponding to the Hilbert space $\mathcal{H}^N$. Thus, for the $N-$partite state $\rho_{\cup_i A_i}$, the irreducible $k-$party quantum correlation ($2\leq k \leq N$) is defined as 
\begin{eqnarray}
\mathscr{C}^{(k)}(\rho_{\cup_i A_i}) &=& S(\rho_{\cup_i A_i}^{(k-1)}) - S(\rho_{\cup_i A_i}^{(k)})~, 
\end{eqnarray}
and the total quantum correlation is given by $\mathscr{C}^T(\rho_{\cup_i A_i})=\sum_{k=2}^N \mathscr{C}^{(k)}(\rho_{\cup_i A_i})$. 

\par\noindent We now define the maximum entanglement entropy state $\tilde{\rho}_{\{\mathcal{A}\}}$ as 
\begin{eqnarray}
\tilde{\rho}_{\mathcal{A}} &=& \textrm{argmax}\{S(\sigma)~|~\sigma\in \mathcal{Q}_{\mathcal{A}}\}~,~\textrm{where}\nonumber\\
\mathcal{Q}_{\mathcal{A}} &=& \{ \sigma ~|~ a\subset \mathcal{A}, |a|=|\mathcal{A}|-1, \sigma_{a}=\rho_{a},  \}~.~~
\end{eqnarray}
The irreducible correlation for the $N-$partite system can now be defined in terms as $\tilde{\rho}_{\{\mathcal{A}\}}$ as follows
\begin{equation}
E_{IC}(\rho_{\cup_i A_i}) = S(\tilde{\rho}_{\cup_i A_i})-S(\rho_{\cup_i A_i})~.~
\label{eq:E_IC}
\end{equation}
For the case of a topologically ordered ground state, and an annular CSS of Fig.\eqref{fig:island_multiplyConnected}(a) under consideration, we can see from relations eq.\eqref{eq:R_k} and eq.\eqref{eq:E_IC} that $\mathscr{C}({\rho}_{\cup_i A_i}^{N-1})=E_{IC}({\rho}_{\cup_i A_i})$. 

\par\noindent In order to proceed towards building a link between the multipartite information $I^{N}$ and irreducible quantum correlations $E_{IC}$, we begin by rewriting the relation eq.\eqref{eq:I_N_recursion} for $I^{N}$ in terms of the total quantum correlation $C^{T}_{N}$ 
\begin{eqnarray}
I^{N}_{\{A\}}&=& \bigg[\displaystyle\sum_{\mu=1}^{N-2}(-1)^{\mu-1} \displaystyle\sum_{\substack{R\in \mathcal{B}_{N-\mu}(\{\mathcal{A}\})\\ |R|=N-\mu}} I^{|R|}_{R} \bigg] + (-1)^{N} \mathscr{C}^T_{N} ~.~~~~~
\label{eq:IN_recurssion_correlation}
\end{eqnarray}
 
For the simple annular CSS structure being considered, one can easily check that only the mutual information ($I^2$) terms for nearest neighbour subsystems will have non-zero values in eq.\eqref{eq:IN_recurssion_correlation}. This can be argued for as follows. Of the ${N\choose 2}$ possible mutual information terms, there are many terms $I^{2}_{\{A_i,A_j\}} = S(\rho_{A_i})+ S(\rho_{A_j}) - S(\rho_{A_i,A_j})$'s where $A_i$ and $A_j$ are not nearest neighbours. For such cases of two disjoint subsystems, one can write the density matrix $\rho_{A_i,A_j}=\rho_{A_i} \otimes \rho_{A_j}$. The mutual information corresponding to $A_i,A_j$ is clearly zero, as $S(\rho_{A_i,A_j})=S(\rho_{A_i})+ S(\rho_{A_j})$~. This leaves us with $N$ number of non-zero mutual information terms $I^{2}_{\{A_i,A_j\}}\neq 0$ (where $A_i$ and $A_j$ are nearest neighbours). Similarly, all $k$-partite information $I^{k}$ with $3\leq k\leq N-1$ must vanish, as no closed loops exist for these sets of $k$ subsystems. Taken altogether, we obtain a simplified expression of $I^N_{\{A\}}$
\begin{eqnarray}
&&I^{N}_{\{\mathcal{A}\}} = (-1)^{N-1} \bigg[\displaystyle\sum_{i=1}^{N} I^2_{\{A_i,A_{i+1\textrm{mod}N}\}} \bigg] +(-1)^N \mathscr{C}^T_N~,~\nonumber\\
&&~~= (-1)^{N-1} \bigg[ \displaystyle\sum_{i=1}^{N} \bigg( S_{A_i}-S_{A_{i}\cup A_{(i+1)\textrm{mod~N }}}  \bigg) + S_{\cup A_i} \bigg]~.~~~~~
\label{eq:IN_with_I2_simplified}
\end{eqnarray}
Using our earlier expression for $I^{N}_{\{\mathcal{A}\}}$ in terms of the TEE ($S_{topo}$), the total correlation $\mathscr{C}^{T}$ is given by
\begin{eqnarray}
\mathscr{C}^{T} &=& \sum_{i=2}^{N}\mathscr{C}^{(i)} ~=\bigg[\displaystyle\sum_{i=1}^{N} I^2_{\{A_i,A_{i+1\textrm{mod}N}\}}\bigg] - \chi S_{topo}~.~
\end{eqnarray}
\begin{figure}[!h]
\includegraphics[scale=0.4]{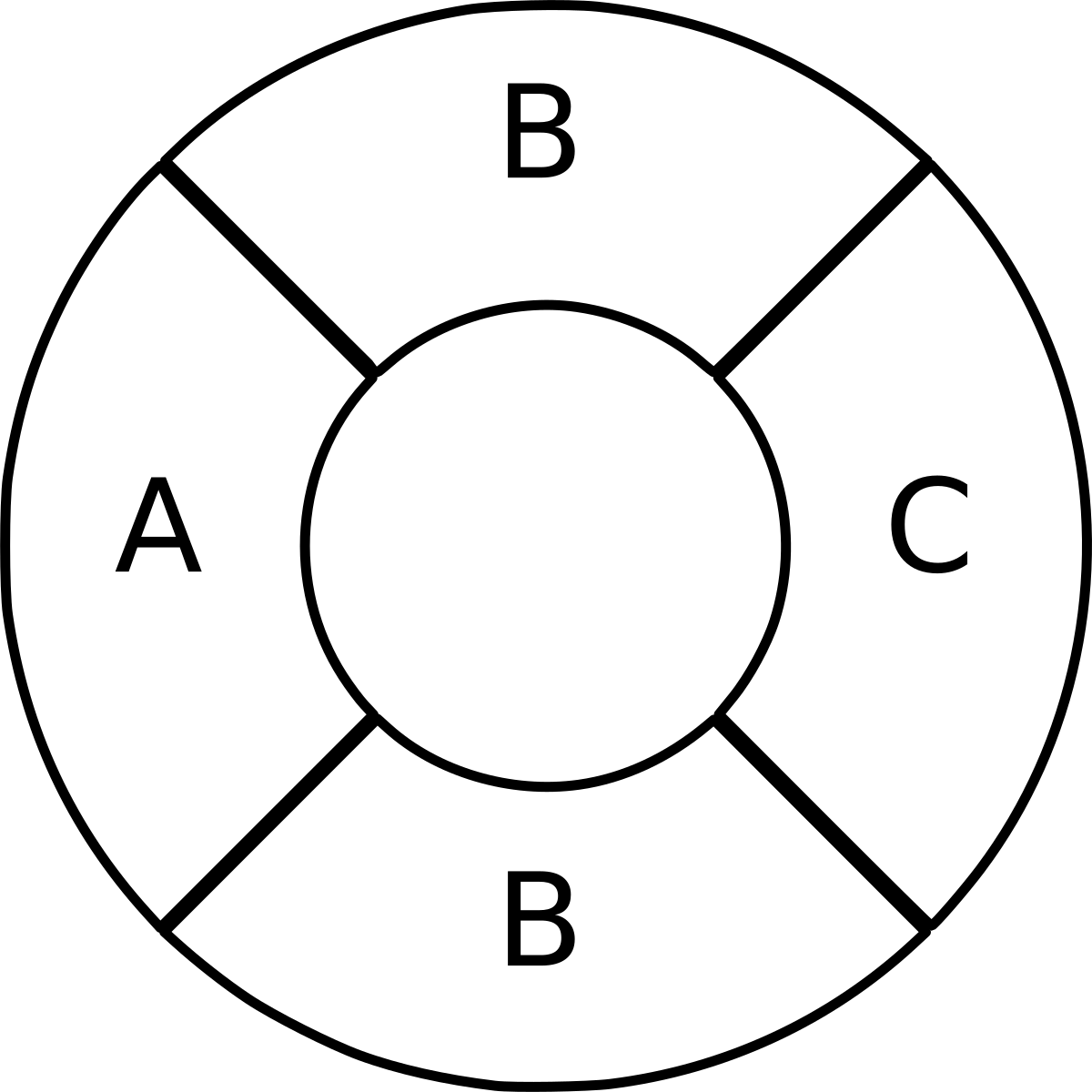}
\caption{CSS ($\{A^{(LW)}_{3}\}$) composed of three subsystems $A,B,C$, and where $B$ is comprised of two disjoint islands~\cite{Levin_Wen_2006}.}
\label{fig:kitaev_annular}
\end{figure}
\par\noindent For the case of the CSS considered in Refs.\cite{Kitaev_Preskill_2006,Levin_Wen_2006} (see Fig.\eqref{fig:kitaev_annular}), we can easily see that from our earlier results that
$I^3_{\{\mathcal{A}^{(LW)}_{3}\}}=-\chi\log \mathcal{D}$. Indeed, this result is in general agreement with the property of strong sub-additivity of von Neumann entanglement entropy for a CSS of $N=3$ subsystems~\cite{Liu_Zhou_2016,Kato_Murao_2016}: 
 $I^3_{\{\mathcal{A}^{(LW)}_{3}\}}\leq 0$~. Similarly, from eq.\eqref{eq:IN_with_I2_simplified}, we obtain a generalized strong sub-additivity relation for a CSS of $N>3$ subsystems in a topologically ordered phase

\begin{eqnarray}
S_{\cup A_i}  +\displaystyle\sum_{i=1}^{N} \bigg( S_{A_i}-S_{A_{i}\cup A_{(i+1)\textrm{mod ~N}}}  \bigg)=-\chi\log\mathcal{D} \leq 0~.~~~~~
\label{eq:strong_sub_add_main}
\end{eqnarray}

The equality in the condition (eq.\eqref{eq:strong_sub_add_main}) corresponds to the case of a topologically trivial phase ($\mathcal{D}=1$, obeying the boundary law entanglement entropy), while the inequality corresponds to the topologically nontrivial phases ($\mathcal{D}>1$). 

\par\noindent
Further, following a similar demonstration for a CSS of $N=3$ subsystems in Ref.\cite{Liu_Zhou_2016}, we obtain for a CSS of $N$ subsystems in a topologically ordered ground state that

\begin{eqnarray}
S(\tilde{\rho}_{\cup_i A_i})   &\leq& \displaystyle\sum_{i=1}^{N} \bigg[S(\tilde{\rho}_{A_{i}\cup A_{(i+1)\textrm{mod~N}}})   -S(\tilde{\rho}_{A_i})\bigg]~,~\nonumber \\
&\leq& \displaystyle\sum_{i=1}^{N} \bigg[S( {\rho}_{A_{i}\cup A_{(i+1)\textrm{mod ~N}}})   -S( {\rho}_{A_i})\bigg]~,~ 
\label{eq:rhotilde}
\end{eqnarray}

\noindent where we have used the fact that $S(\tilde{\rho}_{A_i}) = S( {\rho}_{A_i})$ for any individual subsystem. Now, by subtracting $S(\rho_{\cup_{i}A_{i}})$ from both sides of eq.\eqref{eq:rhotilde} and using eq.\eqref{eq:E_IC}, we obtain 
\begin{eqnarray}
 E_{IC}(\rho_{\cup_{i} A_{i}}) &\leq &\displaystyle\sum_{i=1}^{N} \bigg[S( {\rho}_{A_{i}\cup A_{(i+1)\textrm{mod~N}}})   -S( {\rho}_{A_i})\bigg] -S(\rho_{\cup_i A_i})~,~~~~~ \nonumber \\
 E_{IC}(\rho_{\cup_{i} A_{i}}) &\leq & | I^{N}_{\{\mathcal{A}^{(aD)}_{N}\}} | \equiv \chi S_{topo}~.~ \nonumber
\end{eqnarray}
\par\noindent Thus, we obtain that the $N^{th}$-order irreducible quantum correlation for the choice of subsystems $(\{{A}^{(aD)}_{N}\})$ is bounded from above by the product $\chi S_{topo}$
\begin{eqnarray}
\mathscr{C}^{(N)}(\rho_{\cup_i A_i})= E_{IC}(\rho_{\cup_{i} A_{i}}) &\leq & |I^N_{\{A^{(aD)}_{N}\}}|=\chi S_{topo}~. ~~~\label{eq:allcorr}
\end{eqnarray}
This result is the generalization of the $N=3$ case previously obtained in Refs.\cite{Liu_Zhou_2016,Kato_Murao_2016} for a topologically ordered ground state. 
\par\noindent 
Our results show that, for a non-zero $S_{topo}$, the entanglement Hamiltonian $\tilde{H}_{\rho_{\cup_{i}A_{i}}}\equiv \ln\rho_{\cup_{i}A_{i}}$ on region $\cup_{i}A_{i}$ cannot be a 2-local Hamiltonian \cite{Kato_Murao_2016}. Indeed, $\tilde{H}_{\rho_{\cup_{i}A_{i}}}$ must contain $N$-partite interactions that act on the entire region $\cup_{i}A_{i}$ of the annular CSS. Given that the number of subsystems $N$ is a variable, this suggests that the topologically ordered ground state contains quantum correlations of all orders among subsystems in the form of annular closed loops \cite{Wen_2013,Dora_Moessner_2018}. This is likely to be particularly relevant to the nature of the entanglement encoded in the $Z_{2}$ topologically ordered string loop condensed phases of models like the toric code \cite{Kitaev2006,ClaudioClaudio2007}.

\section{Discussion}
\label{section:discussion}

\par\noindent Topological entanglement entropy (denoted by $S_{topo}$) is a property unique to a topologically ordered system, and arises from the quantum dimension of its degenerate ground state manifold. In order to extract the TEE, we rely on an entanglement measure (e.g., multipartite information) that does not depend on the geometry of the subsystems employed in the measurement. Such an entanglement measure based on tripartite information ($I^{3}$) was proposed in Refs.\cite{Kitaev_Preskill_2006,Levin_Wen_2006}. Here, we have generalised this measure to the multipartite information ($I^{N}$) for an annular arrangement of $N$ subsystems. This has then helped unveil the dependence of $I^{N}$ on the topology of such an annular collection of the subsystems (CSS). Specifically, for all $N\geq 3$, we find that $I^{N}$ is a topological invariant given simply by $|I^{N}| = \chi S_{topo}$~, where $\chi (\equiv 2)$ is the Euler characteristic of the CSS embedded on the planar manifold.

\par\noindent We have also analysed carefully the robustness of $I^{N}$ to changes in the topology of the CSS from a simple annular structure: neither the inclusion of self-handles (or holes) within individual subsystems, nor handles between nearest neighbour subsystems, changes our result for $I^{N}$. While the inclusion of handles between subsystems beyond nearest neighbour causes $I^{N}$ to vanish, it becomes possible to identify similar multipartite information measures for several smaller annular CSS that again extract $\chi S_{topo}$. Thus, we conclude that one can very generally construct a simple annular structure of $N\geq 3$ subsystems, such that their $I^{N}$ can unambiguously capture $S_{topo}$. Further, we have also shown that for any complex CSS structure containing $n_{h}$ number of holes, the sum of the individual multipartite informations measured around each of the holes is constrained by the product $n_{h}\chi S_{topo}$.

\par\noindent  Further, we believe that our finding of an identical value of $I^{N}$ for all annular structures with $N\geq 3$ indicates the special nature of topologically ordered ground states. In 
order to quantify this, we define a $N-2$-component vector comprising the various $I^{N}$ ($N\geq 3$) multipartite informations as follows
\begin{equation}
\hat{\mathcal{E}}_{N}=~\mathcal{N}(|I^3_{\{A_3\}}|,|I^4_{\{A_4\}}|,~.~.~,|I^{N}_{\{A_{N}\}}|)~,
\end{equation}
where the normalisation factor $\mathcal{N} = {\sqrt{\frac{N-2}{\sum_{p=3}^{N}|I^p_{\{A_p\}}|^2 } }}$~.
We propose that $\hat{\mathcal{E}}_{N}$ can be used to classify various phases in terms of their multipartite entanglement content, as well as the phase transitions among them. For instance, we expect that metallic phases should be represented by $\hat{\mathcal{E}}_M=(0,0,~.~.~,0)$, i.e., the origin of $N-2$-dimensional multipartite information space. On the other hand, topologically ordered phases have been shown by us to be represented by the point $\hat{\mathcal{E}}_M=(1,1,~.~.~,1)$. It will be interesting to see where other phases lie within this unit hypercube.
\par\noindent Our investigations have also revealed that the $N$-party irreducible quantum correlations among the $N$ subsystems of a annular arrangement is bounded from above by $\chi S_{topo}$ for any $N$. The independence of this result on $N$ provides evidence of the fact that closed loop-like structures of all sizes are present within the ground state of a topologically ordered system. We believe that this is of relevance to understanding the patterns of entanglement encoded within the string loop condensed phases of topological quantum matter (see, e.g., Ref.\cite{Wen_2013} and references therein). It will also be an interesting challenge to extend these ideas to the topologically ordered phases that have been recently proposed by some of us in strongly correlated electronic (e.g., Mott liquid, Cooper pair insulator~\cite{MukherjeeMott1,MukherjeeMott2,MukherjeeNPB2,1dhubjhep,siddharthacpi}) and quantum spin systems in frustrated lattice geometries at finite fields~\cite{pal2019magnetization,pal2019,pal2020}.\\

\acknowledgments
\noindent The authors thank  D. Bhasin, K. Sinha, M. Podder and S. Bhattacharya for several discussions and feedback. S. P. thanks the CSIR, Govt. of India and IISER Kolkata for funding through a research fellowship. S.B. acknowledges SERB Matrics grant MTR/2017/000807 for the funding opportunity.

\appendix

\begin{widetext}

\section{The case of cycle graphs} 
\label{ap:cycle_graph}

\begin{figure}[!h]
\includegraphics[scale=0.28]{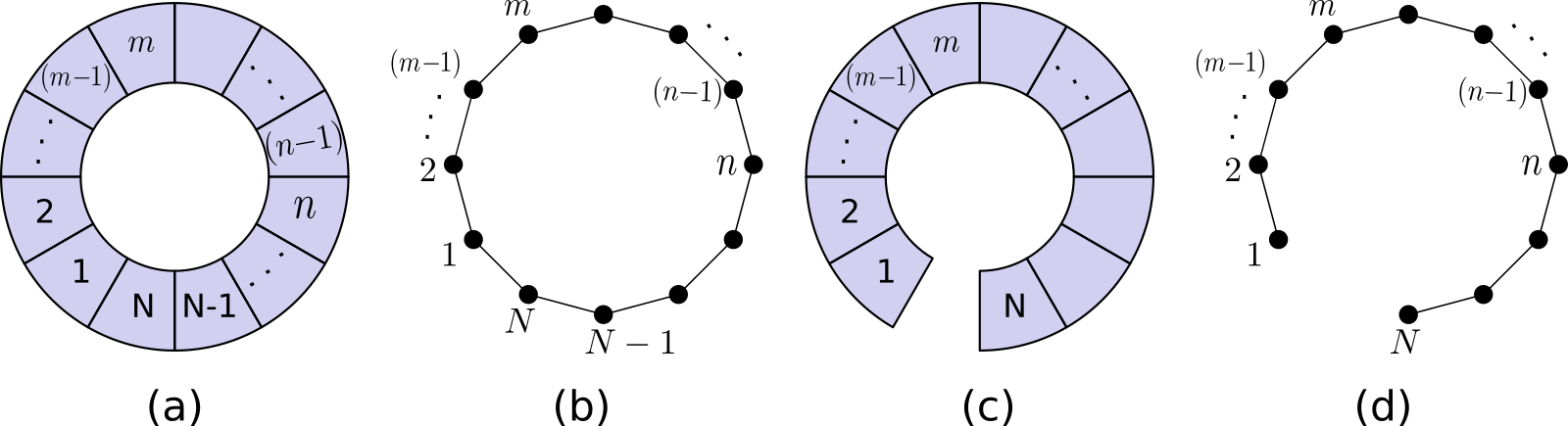}
\caption{Figures displaying the graphs equivalent to various CSS topologies. (a) The simple annular CSS ($\{\mathcal{A}^{(a)}_{N}\}$) with $N$ subsystems. (b) Graph ($\Gamma=\Upsilon(\{\mathcal{A}^{(a)}_{N}\})$) corresponding to the CSS ($\{\mathcal{A}^{(a)}_{N}\}$, Fig.(a)), where each subsystem $A_i$ is replaced by a node $i$ and the wall between two subsystems $A_i,A_{i+1}$ is replaced by an edge. $\Gamma=\Upsilon(\{\mathcal{A}^{(a)}_{N}\})$ has a total of $N$ nodes, as well as $N$ edges. (c) An open CSS ($\{\mathcal{A}^{(o)}_N\}$) with $N$ subsystems. (d) Graph ($\Gamma=\Upsilon(\{\mathcal{A}^{(o)}_N\})$) corresponding to the CSS $\{\mathcal{A}^{(o)}_N\}$, containing $N$ number of nodes and $N-1$ number of edges.}
\label{fig:ribbon_sceleton}
\end{figure}

\par\noindent 
We recall the definition $\mathcal{C}^N_{\{\mathcal{A}_N\}}$ for the CSS $\{\mathcal{A}_N\}$ where there is no overlap among all the $N$ subsystems
\begin{eqnarray}
\mathcal{C}^N_{\{\mathcal{A}_N\}}&=& \displaystyle\sum_{m=1}^{N} (-1)^{m-1} \displaystyle\sum_{Q \in \mathcal{B}_m(\{\mathcal{A}_N\})} {\mathcal{J}}_{V_{\cup}({Q})} \nonumber\\
&=& \Sigma_{\{\mathcal{A}_N\}} + (-1)^{N-1} \mathcal{J}_{V_{\cup}(\{\mathcal{A}_N\})}~,~\nonumber\\
\textrm{where,}~~ \Sigma_{\{\mathcal{A}_N\}} &=& \displaystyle\sum_{m=1}^{N-1} (-1)^{m-1} \displaystyle\sum_{Q \in \mathcal{B}_m(\{\mathcal{A}_N\})} {\mathcal{J}}_{V_{\cup}({Q})} ~,~~~~\nonumber
\label{eq:ap_CN}
\end{eqnarray}
where $\mathcal{J}_A$ represents the number of disconnected boundaries of the subsystem $A$. We define a graph $\Gamma$ corresponding to a CSS $\{\mathcal{A}_N\}$, $\Gamma=\Upsilon(\{\mathcal{A}_N\})$. Each subsystem $A_i$ in the CSS is replaced by a node ($i$), and each connectivity between two subsystems $A_i$ and $A_j$ is replaced by edges between corresponding two nodes $i$ and $j$ (e.g., Fig.\eqref{fig:ribbon_sceleton}. A graph is denoted by $\Gamma=(V(\Gamma), E(\Gamma))$, where $V$ is the set of vertices and $E$ is the set of edges. Let $v=|V|$ and $e=|E|$ denote the number of vertices and edges respectively. We shall now deal with subgraphs. In particular, recall that a subgraph $\Gamma'$ with a vertex set $S$ is called induced if any edge in $\Gamma$ joining two vertices in $S$ is also in the subgraph. We will be typically be dealing with nontrivial induced subgraphs, i.e., an induced subgraph where the vertex set is neither $\emptyset$ nor $V(\Gamma)$.\\

\begin{defn}
For a finite graph $\Gamma$, we define the integer
\begin{displaymath}
\rho(\Gamma):=\sum_{i=1}^{v-1} (-1)^i\sum_{\mathscr{F}_i}  H_0(\Gamma'),
\end{displaymath}
where $\mathscr{F}_i$ contains all induced subgraphs $\Gamma'$ of $\Gamma$ such that $v(\Gamma')=i$. The integer $H_0(\Gamma')$ is the number of connected components of $\Gamma'$. 
\end{defn}
\noindent We shall use $\mathscr{F}(\Gamma)$ to denote the collection of nontrivial induced subgraphs of $\Gamma$.  Calculating the number of disconnected boundaries of a subsystem $A$ is equivalent of calculating the number of connected components in the graph corresponding to the subsystem $A$. To be exact the relation between the number of disconnected boundaries of the subsystem and the number of connected components of the corresponding graph is given as,
\bgd
\displaystyle\sum_{Q \in \mathcal{B}_m(\{\mathcal{A}_N\})} {\mathcal{J}}_{V_{\cup}({Q})}=\sum_{\mathscr{F}_m}  H_0(\Gamma')~,~~\forall m<N
\edd
Thus we can see from the definitions above $\Sigma_{\{\mathcal{A}_N\}}=-\rho(\{\mathcal{A}_N\})$. Here we are interested in calculating $\Sigma_{\{\mathcal{A}_N\}}$ for a CSS $\{\mathcal{A}_N\}$.
\subsection{$\mathcal{C}^{N}_{\{\mathcal{A}^{(o)}_N\}}=0$, for an open structured CSS $\{\mathcal{A}^{(o)}_N\}$}
\label{ap:path_graph}
\pin Here we are interested in calculating $\mathcal{C}^{N}_{\{A^{(o)}_N\}}$. The graph corresponding to the CSS  ${\{A^{(o)}_N\}}$ is $\Upsilon(\{A^{(o)}_N\})=P_N$, i.e., a path graph with $N$ number of nodes.

\begin{prop}\label{inv-path}
For the path graph $P_m$ on $m\geq 3$ vertices, the invariant $\rho(P_m)=(-1)^{m-1}$. 
\end{prop}
\begin{proof}
A path graph $P_n$ has $n$ vertices and $n-1$ edges. As a path graph is contractible, i.e., homotopy equivalent to any vertex, it follows that $H_0(P_n)=1$. Any induced subgraph $\Gamma$ is a disjoint union of path graphs. Therefore, if $\Gamma=\Gamma_1\sqcup \cdots\sqcup \Gamma_r$, then $H_0(\Gamma)=r$. We use induction to compute $\rho(P_n)$. Observe that $P_{n+1}$ is constructed from $P_n$ by adding an extra vertex labelled $n+1$ and an edge $e$ joining vertex $n$ to $n+1$. \\
\hf Notice that the nontrivial induced subgraphs $\Gamma$ of $P_{n+1}$ (for $n\geq 2$) are of three mutually exclusive and exhaustive types:\\
(a) $n+1\not\in V(\Gamma)$: These are actually induced subgraphs of $P_n$, including $P_n$ itself. \\
(b) $n, n+1\in V(\Gamma)$: These graphs are obtained from nontrivial induced subgraphs $\Gamma'$ of $P_n$ by adjoining $e$. Thus, $H_0(\Gamma)=H_0(\Gamma')$. \\
(c) $n+1\in V(\Gamma)$ but $n\not\in V(\Gamma)$: These graphs are obtained from induced subgraphs $\Gamma'$ of $P_{n-1}$, including $P_{n-1}$ and $\emptyset$, by adjoining the vertex $n+1$. Thus, $H_0(\Gamma)=H_0(\Gamma')+1$. \\
The invariant for $P_{n+1}$ can be computed from the three types of contributions as follows. Type (a) contributes $\rho(P_n)$ which is the sum of three quantities: \\
\hf $\alpha$ - the contribution from $\mathscr{F}(P_{n-1})$;\\
\hf $\beta$ - the contribution from $\mathscr{F}(P_n)$ containing vertex $n$;\\
\hf $(-1)^{n-1}$ - the contribution from $P_{n-1}$ itself;\\
\hf $(-1)^n$ - the contribution from $P_n$ itself.\\
As type (b) contributes $-\beta$, the total contribution from types (a) and (b) is $\alpha$. Type (c) contributes 
\begin{eqnarray}
&&-1+\sum_{i=1}^{n-1}  (-1) ^{i+1}\sum_{\Gamma'\in \mathscr{F}_i(P_{n-1})} (H_0(\Gamma')+1)\nonumber\\
&&=-1-\sum_{i=1}^{n-1}(-1)^i \sum_{\Gamma'\in \mathscr{F}_i(P_{n-1})} H_0(\Gamma')   -\sum_{i=1}^{n-1} (-1)^i|\mathscr{F}_i(P_{n-1})|\nonumber\\
&&~~~~ = -1+(-1)^n -\alpha-\sum_{i=1}^{n-1} (-1)^i {n-1 \choose i} \nonumber\\
&&~~~~ = (-1)^n-\alpha-\sum_{i=0}^{n-1} (-1)^i {n-1 \choose i} \nonumber\\
&&~~~~ = (-1)^n-\alpha.
\end{eqnarray}
Thus, $\rho(P_{n+1})=(-1)^n$, being the sum of contributions from (a), (b) and (c). Thus we find $\rho(P_n)=(-1)^{n-1}$.
\end{proof}

\pin Using the above relation we find that for an open structured CSS $\{\mathcal{A}^{(o)}_N\}$ as shown in the Fig.\ref{fig:ribbon_sceleton}(c,d),  $\Sigma_{\{\mathcal{A}^{(o)}_N\}}=(-1)^{N-1}$. Thus
\begin{eqnarray}
\mathcal{C}^{N}_{\{A^{(o)}_N\}}&=&\Sigma_{\{\mathcal{A}^{(o)}_N\}} + (-1)^{N-1}\mathcal{J}_{V_{\cup}(\{\mathcal{A}^{(o)}_{N}\})} ~,~\nonumber\\
&=& -\rho(P_N) + (-1)^{N-1}~,\nonumber\\
&=& (-1)^N+ (-1)^{N-1} = 0~,
\end{eqnarray}
Thus, it is proved that for an open structured CSS that the count $\mathcal{C}^N_{\{\mathcal{A}^{(o)}_N\}}=0$. Therefore, the multipartite information measure for this particular choice of CSS is given by $I^N_{\{\mathcal{A}^{(o)}_N\}}=-\mathcal{C}^N_{\{\mathcal{A}^{(o)}_N\}} \log \mathcal{D}=0$.

\subsection{$\Sigma_{\{\mathcal{A}^{(a)}_N\}}= 0$ for an closed structured CSS $\{\mathcal{A}^{(a)}_N\}$}
\label{ap:cycle_graph1}
\pin We now calculate $\Sigma_{\{\mathcal{A}^{(a)}_N\}}$ for the closed annular structured CSS $\{\mathcal{A}^{(a)}_N\}$. The corresponding graph is $\Upsilon(\{\mathcal{A}^{(a)}_N\})=C_N$, i.e., the cycle graph with $N$ number of vertices and nodes.

\begin{cor}
For the cycle graph $C_n$, we have $\rho(C_n)=0$.
\end{cor}

\begin{proof}
Recall that the cycle graph $C_n$ is a graph on $n$ vertices and $n$ edges (Fig.\ref{fig:ribbon_sceleton}(b)), such that each vertex has valency two. This graph is usually visualized as the boundary of a regular $n$-gon. Observe that $C_n$ is obtainable from $P_n$ by attaching an edge $e$ joining $1$ and $n$. The induced subgraphs $\Gamma$ of $P_{n}$ (for $n\geq 3$) are of three mutually exclusive and exhaustive types:\\
\hf (a) $1,n\not\in V(\Gamma)$: let $\alpha$ be the contribution from these towards $\rho(P_n)$;   \\
\hf (b) exactly one of $1$ and $n$ is in $V(\Gamma)$: let $\beta$ be the contribution from these towards $\rho(P_n)$;   \\
\hf (c) $1,n\in V(\Gamma)$: let $\gamma$ be the contribution from these towards $\rho(P_n)$.  \\
In particular, we have $\alpha+\beta+\gamma=\rho(P_n)=(-1)^{n-1}$. Now note that type (a) and (b) are induced subgraphs of $C_n$; the contribution of these types towards $\rho(C_n)$ will be $\alpha$ and $\beta$ respectively. The other induced subgraphs are modifications of those in (c) - we have to add the edge $e$ in order to type (c) subgraphs. Adding an edge decreases the number of connected components by $1$, whence 
\begin{eqnarray*}
\tilde{\gamma} & = & \sum_i \sum_{\{\Gamma\subset P_n\,|\,1,n\in V(\Gamma),v(\Gamma)=i\}} (-1)^i (H_0(\Gamma)-1)\\
& = & \sum_i \sum_{\{\Gamma\subset P_n\,|\,1,n\in V(\Gamma),v(\Gamma)=i\}} (-1)^i  H_0(\Gamma) -\sum_i \sum_{\{\Gamma\subset P_n\,|\,1,n\in V(\Gamma),v(\Gamma)=i\}} (-1)^i \\
& = & \gamma-\sum_{i=2}^{n-1} (-1)^i {n-2 \choose i-2}\\
& = & \gamma+(-1)^{n-2}.
\end{eqnarray*}
Adding $\alpha,\beta$ and $\tilde{\gamma}$ we obtain $\rho(C_N)=\alpha+\beta+\tilde{\gamma}=0$. Thus $\Sigma_{\{\mathcal{A}^{(a)}_{N}\}}=-\rho(C_N)=0$.
\end{proof}

\section{Simple annular structure and Euler characteristic}
\label{ap:simpleannular_euler}
\par\noindent 
For a simple annular structure of CSS $(\{\mathcal{A}^{(a)}_N\})$ shown in the Fig.\ref{fig:island_subsystem_main}(a), we can calculate the multipartite information by using eq.\eqref{eq:I_N_recursion} as follows
\begin{eqnarray}
I^{N}_{\{\mathcal{A}^{(a)}_N\}}&=& \bigg[\sum_{\substack{j\\ \mu_j\geq 3}} (-1)^{(N-1)-\mu_j} I^{\mu_j}_{\{\mathcal{A}_{\mu_j}\}}\bigg] + \bigg[ (-1)^{(N-1)-2} \sum_{\{A_i,A_j\}\in \{M\}} I^2_{A_i,A_j}\bigg]+ (-1)^{N} \bigg[\displaystyle\sum_{i} S_{A_i} - S_{\cup A_i}  \bigg] ~.~ 
\end{eqnarray}
For such a simple annular structure, we obtain a vanishing multipartite information for all CSS composed of $m(<N)$ number of subsystems that do not form closed loops. Thus, $\sum_{\substack{j\\ \mu_j\geq 3}} (-1)^{(N-1)-\mu_j} I^{\mu_j}_{\{A_j\}} =0$, and we obtain
\begin{eqnarray}
I^{N}_{\{\mathcal{A}^{(a)}_N\}}&=&(-1)^{N} \bigg[ d_{nn}-N +n_h+1  \bigg] \log\mathcal{D} ~. ~ 
\end{eqnarray}
For the simple annulus, $N=d_{nn}$, the number of edges is $e=N$, the number of vertices is $v=N$ and the number of faces is $f=n_h+1$. This leads to $\chi=e-v+f=n_h+1=2$ (confirming the value of the Euler characteristic for the planar manifold on which the annulus is embedded). Thus, we can write the multipartite information measure for simple annular structure as
\begin{eqnarray}
I^{N}_{\{\mathcal{A}^{(a)}_N\}}&=&(-1)^{N} \bigg[ d_{nn}-N +n_h+1  \bigg] \log\mathcal{D} ~ ~ \nonumber\\
&=& (-1)^N \chi \log\mathcal{D}= (-1)^N \chi S_{topo}~.
\end{eqnarray}

\section{Isolated structure}
\label{ap:isolated_Structure}
\par\noindent 
We now turn to the case of a CSS that does not form a closed loop, $\{\mathcal{A}^{(1b)}_{N}\}=\{\mathcal{A}^{(1b),N}_{N-1}\}\cup \{A_N\}$. Using eq.\eqref{eq:N-partite-connectivity}, one can easily see that

\begin{eqnarray}
&&\mathcal{C}^N_{\{\mathcal{A}^{(1b)}_N\}} = \mathcal{C}^{N-1}_{\{\mathcal{A}^{(1b)}_N\}-\{A_N\}}+\mathcal{J}_{A_N} -\displaystyle\sum_{i=1}^{N-1} \mathcal{J}_{A_N\cup A_i}~~~+ \displaystyle\sum_{i<j=1}^{N-1} \mathcal{J}_{A_N\cup A_i\cup A_j} \cdots +(-1)^{N-1}\mathcal{J}_{A_N\cup A_1\cup ~\cdots~ \cup A_{N-1}}~.
\end{eqnarray}
\noindent Using the fact that subsystem $A_N$ is disjoint from the rest of the system, we can see that $\mathcal{J}_{A_N\cup \{\mathcal{B}\}}=\mathcal{J}_{A_N}+\mathcal{J}_{\{\mathcal{B}\}}$. Thus, we can simplify the above equation as follows
\begin{eqnarray}
\mathcal{J}_{A_N} &&-\displaystyle\sum_{i=1}^{N-1} \mathcal{J}_{A_N\cup A_i}~~~+ \displaystyle\sum_{i<j=1}^{N-1} \mathcal{J}_{A_N\cup A_i\cup A_j}+ \cdots +(-1)^{N-1}\mathcal{J}_{A_N\cup A_1\cup ~\cdots~ \cup A_{N-1}}\nonumber\\
&=& \mathcal{J}_{A_N} -\displaystyle\sum_{i=1}^{N-1} (\mathcal{J}_{A_N}+\mathcal{J}_{A_i})~~~+ \displaystyle\sum_{i<j=1}^{N-1} (\mathcal{J}_{A_N}+\mathcal{J}_{A_i\cup A_j})+ \cdots +(-1)^{N-1}(\mathcal{J}_{A_N}+\mathcal{J}_{A_1\cup ~\cdots~ \cup A_{N-1}})\nonumber\\
&=& \mathcal{J}_{A_N} \displaystyle\sum_{i=0}^{N-1} (-1)^{i} {N-1 \choose i}  -\displaystyle\sum_{i=1}^{N-1} \mathcal{J}_{A_i}~~~+ \displaystyle\sum_{i<j=1}^{N-1} \mathcal{J}_{A_i\cup A_j} +\cdots +(-1)^{N-1}\mathcal{J}_{A_1\cup ~\cdots~ \cup A_{N-1}}\nonumber\\
&=&  \mathcal{J}_{A_N}\Upsilon_{N-1} -\mathcal{C}^{N-1}_{\{\mathcal{A}^{(1b)}_N\}-\{A_N\}}~,
\end{eqnarray}
where $\Upsilon_{N-1}=\displaystyle\sum_{i=0}^{N-1} (-1)^{i} {N-1 \choose i}=0$. In turn, we obtain
\begin{eqnarray}
\mathcal{C}^N_{\{\mathcal{A}^{(1b)}_N\}}&&= \mathcal{C}^{N-1}_{\{\mathcal{A}^{(1b)}_N\}-\{A_N\}}-\mathcal{C}^{N-1}_{\{\mathcal{A}^{(1b)}_N\}-\{A_N\}}+\mathcal{J}_{A_N} \Upsilon_{N-1} = 0~.~
\end{eqnarray}
Hence, the multipartite information measure for this CSS is seen to vanish: $I^N_{\{\mathcal{A}^{(1b)}_N\}}=-\mathcal{C}^N_{\{\mathcal{A}^{(1b)}_N\}} \log \mathcal{D}=0$.

\section{Annular structure with appendage}
\label{ap:Appendage}
\par\noindent
As shown in Fig.\ref{fig:island_subsystem_main}(d), we consider here a CSS with $N$ number of subsystems and containing an appendage (the $N$th subsystem). In order to compute the multipartite information for this CSS, $I^N_{\{\mathcal{A}^{(1d)}_{N}\}}$, we use eq.\eqref{eq:I_N_recursion} 
\begin{eqnarray}
I^{N}_{\{\mathcal{A}^{(1d)}_{N}\}}&=& \displaystyle\sum_{\mu=1}^{N-2}(-1)^{\mu-1} \displaystyle\sum_{\substack{R\in \mathcal{B}_{N-\mu}(\{\mathcal{A}^{(1d)}_{N}\})\\ |R|=N-\mu}} I^{|R| }_{R}+ (-1)^{N} \bigg(  \displaystyle\sum_{i} S_{A_i} - S_{\cup A_i}  \bigg)~.
\end{eqnarray}
We can see that except for $I^{N-1}_{\{A_1,~.~.~, A_{N-1}\}}$, all the terms $I^m_{\{..\}}$ for $2<m<N$ are zero. This is because they either form an open line, or composed of isolated islands. Further, we have already shown that the CSS of an open line structure, or one composed of isolated islands, gives a vanishing multipartite information. Thus, the above equation reduces to 
\begin{eqnarray}
I^{N}_{\{\mathcal{A}^{(1d)}_{N}\}}&=& (-1)^{N-1}I^3_{\{\mathcal{A}^{(1d)}_N\}} +(-1)^{N-1} \displaystyle\sum_{\substack{R\in \mathcal{B}_{2}(\{\mathcal{A}^{(1d)}_{N}\})\\ |R|=2}} I^{|R| }_{R}+ (-1)^{N} \bigg(  \displaystyle\sum_{i} S_{A_i} - S_{\cup A_i}  \bigg)~. ~
\end{eqnarray}
Now, we know that $I^2_{A_i,A_j}=0$ if $\mathcal{J}_{A_i}=\mathcal{J}_{A_j}=1$ and $\mathcal{J}_{A_i\cup A_j}=2$ $\forall i,j$. This shows that when two subsystems are not touching each other, i.e., they are disjoint, their joint density matrix can be decomposed into a product form: $\rho_{A_i\cup A_j}=\rho_{A_i}\otimes \rho_{A_j}$. Using this for the case of Fig.\ref{fig:island_subsystem_main}(d), we obtain
\begin{eqnarray}
I^{N}_{\{\mathcal{A}^{(1d)}_{N}\}}&=& (-1)^{N-1}(-1)^{3}\log\mathcal{D}^2 +(-1)^{N-1} \bigg[-N\log\mathcal{D}\bigg]+ (-1)^{N} \bigg( -N\log\mathcal{D} +2\log\mathcal{D}  \bigg) ~ \nonumber\\
&=& (-1)^{N-1}\bigg[ -2-N-(-N+2) \bigg] \log\mathcal{D}=0~.~
\end{eqnarray}
The above result shows that adding an appendage subsystem to a simple annular structure trivializes the computation of the multipartite information $I^{N}$, and is unable to capture the topological entanglement entropy.

\section{Many holes in the CSS}
\label{ap:holes_in_css}

\begin{figure}[!h]
\includegraphics[scale=0.6]{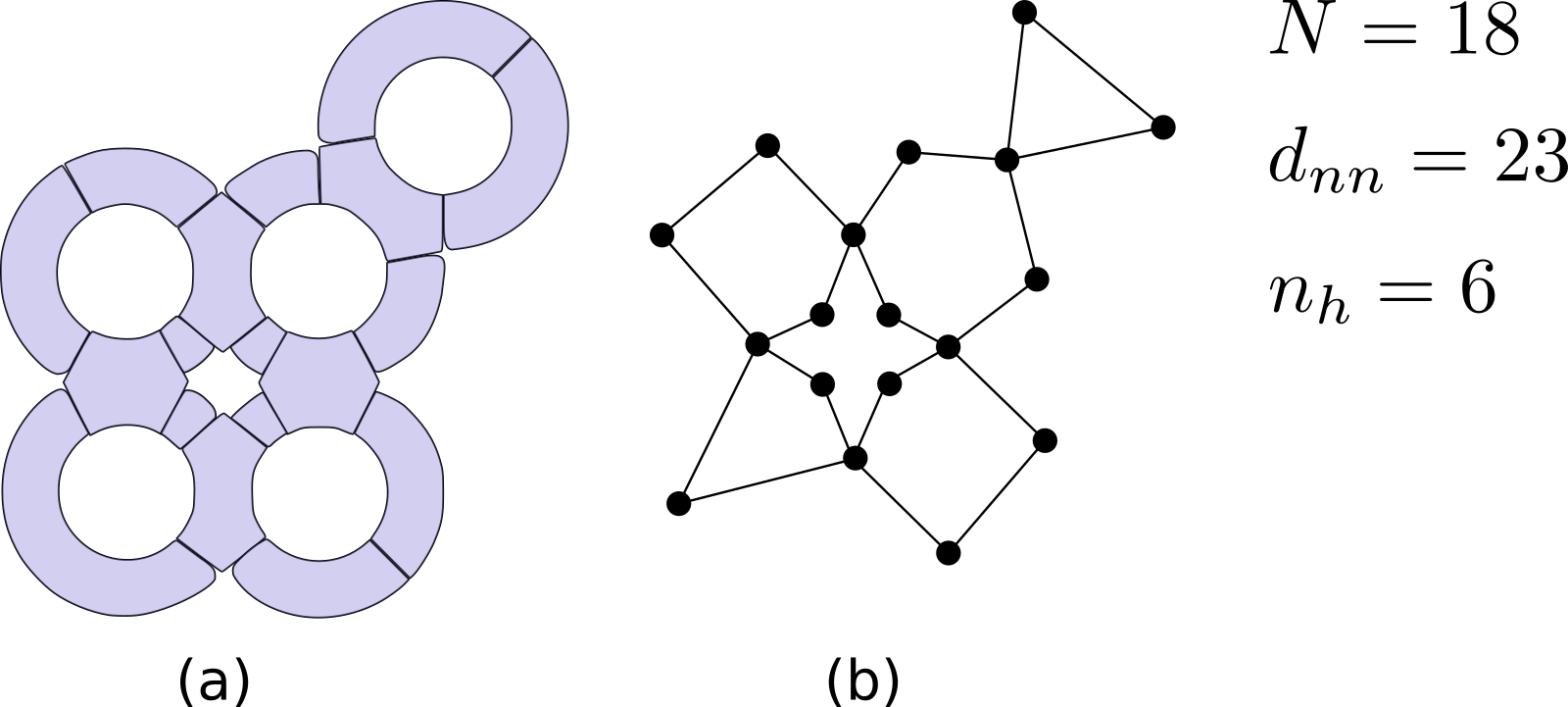}
\caption{In general, a CSS has $n_h$ number of holes, $N$ subsystems and$d_{nn}$ partitions. {(a)} An example of a CSS ($\{\mathcal{A}^{(6)}_{N}\}$) with $n_h=6$. {(b)} Graph corresponding to $\{\mathcal{A}^{(6)}_{N}\}$, where each vertex represents a subsystem and each edge represents a partition.}
\label{fig:genus-6}
\end{figure}
\par\noindent 
We now discuss the case where the $N$ subsystems are arranged in such a way that the CSS has $n_h$ number of holes, denote as $\{\mathcal{A}^{(n_h)}_N\}$. An example is given in Fig.\eqref{fig:genus-6}, where the CSS has $n_h=6$ number of holes. As before, we are interested in calculating the multipartite information $I^N_{\{\mathcal{A}^{(n_h)}_N\}}$ by using eq.\eqref{eq:I_N_recursion}. Here, we are taking the simple case where an individual subsystem has a single disconnected boundary $\mathcal{J}_{A_i}=1~, \forall A_i\in\{\mathcal{A}^{(n_h)}_N\}$. There are $d_{nn}$ number of pairs or subsystems ($A_i,A_j$) where $\mathcal{J}_{A_i\cup A_j}=\mathcal{J}_{A_i}=\mathcal{J}_{A_j}=1$. Thus, we obtain
\begin{eqnarray}
I^{N}_{\{\mathcal{A}^{(n_h)}_N\}}&=& \bigg(\displaystyle\sum_{\mu=1}^{N-2}(-1)^{\mu-1} \displaystyle\sum_{\substack{R\in \mathcal{B}_{N-\mu}(\{\mathcal{A}\})\\ |R|=N-\mu}} I^{|R| }_{R} \bigg) + (-1)^{N} \bigg(  \displaystyle\sum_{i} S_{A_i} - S_{\cup A_i}  \bigg) ~. ~~ 
\label{eq:ap_id_11}
\end{eqnarray}
This relation shows that $I^{N}$ is comprised of many different multipartite information terms that differ in the numbers of subsystems involved. From our earlier discussions, the only nontrivial multipartite informations are those that correspond to an annular CSS. Now, one can create an annular CSS (formed out of  say $\mu_j$ number of subsystems) around each hole ($j$); we denote these CSS as $\{A_i\}$. Then, $I^{\mu_j}_{\{A_j\}}= (-1)^{\mu_j} \log\mathcal{D}^2$. Similarly, the only nontrivial mutual informations are those where both subsystems are touching one another: $I^2_{A_i,A_j}=-\log\mathcal{D}, \forall i,j$ if $\mathcal{J}_{A_i\cup A_j}=\mathcal{J}_{A_i}=\mathcal{J}_{A_j}=1$, and we represent the set of all such pairs of subsystems as $\{M\}$ (with cardinality $|\{M\}|=d_{nn}$). Using this rule, we can obtain the $n_h$ number of non-zero multipartite informations in the above eq.\eqref{eq:ap_id_11}
\begin{eqnarray}
I^{N}_{\{\mathcal{A}^{(n_h)}_N\}}&=& \bigg[\sum_{j} (-1)^{(N-1)-\mu_j} I^{\mu_j}_{\{A_j\}}\bigg] + \bigg[ (-1)^{(N-1)-2} \sum_{\{A_i,A_j\}\in \{M\}} I^2_{\{A_i,A_j\}}\bigg]+ (-1)^{N} \bigg[\displaystyle\sum_{i} S_{A_i} - S_{\cup A_i}  \bigg] ~.~ 
\label{eq:ap_1246}
\end{eqnarray}
We focus on the case where $\mathcal{J}_{A_i}=1~, \forall i$, the topological part of $S_{A_i}$ is $-\log \mathcal{D}$, and the corresponding topological part of $S_{\cup_iA_i}$ is $-(n_h+1)\log\mathcal{D}$ (as it has $n_h+1$ number of disconnected boundaries). Thus, we can further simplify the above relation as
\begin{eqnarray}
I^{N}_{\{\mathcal{A}^{(n_h)}_N\}}&=& \bigg[\sum_{j}^{n_h} (-1)^{(N-1) }   2\log\mathcal{D}\bigg] + \bigg[ (-1)^{(N-1)-2}(-1)  d_{nn} \log \mathcal{D}\bigg]+ (-1)^{N} \bigg(  -N\log\mathcal{D} + (n_h+1)\log \mathcal{D}  \bigg) ~ ~~ \nonumber\\
&=& \bigg[  (-1)^{(N-1) }   2{n_h}\log\mathcal{D}\bigg] + \bigg[ (-1)^{N }  d_{nn} \log \mathcal{D}\bigg]+ (-1)^{N} \bigg[  -N\log\mathcal{D} + (n_h+1)\log \mathcal{D}  \bigg] ~ ~~ \nonumber\\
&=& (-1)^{(N-1)} \bigg[ {n_h} - d_{nn}+N  - 1 \bigg]\log\mathcal{D}=(-1)^{(N-1)} \bigg( \chi  - 2 \bigg)\log\mathcal{D} ~, ~~ 
\label{eq:ap_id_22}
\end{eqnarray}
where $\chi$ is the Euler characteristic of the underlying spatial manifold. As this manifold is planar in our case, we know that $\chi=2$. Hence, the above equation vanishes very generally. We can also easily verify that this relation vanishes for the specific case shown in Fig.\eqref{fig:genus-6}: $N=18$, $d_{nn}=23$, $n_h=6$, giving $N-d_{nn}+n_h-1=0=I^{N}_{\{\mathcal{A}^{(6)}_N\}}$.

\section{Recursion in multipartite information}\label{app:recursionIN}
\noindent Our goal here is to prove very generally the following relation 
\begin{eqnarray}
I^{N}_{\{\mathcal{A}_N\}} =  \displaystyle\sum_{\mu=1}^{N-2}(-1)^{\mu-1} \displaystyle\sum_{\substack{R\in \mathcal{B}_{N-\mu}(\mathcal{A})\\ |R|=N-\mu}} I^{|R| }_{R}+ (-1)^{N} \bigg(  \displaystyle\sum_{i} S_{A_i} - S_{\cup A_i}  \bigg).
\label{eq:tobeproved_master}
\end{eqnarray}
This equation shows the relation of $N-$partite information with various lower-order multipartite informations. Using the fact that $I^1_{A_i}=S_{A_i}$, we can re-write the above equation as 
\begin{eqnarray}
I^{N}_{\{\mathcal{A}_N\}} &=&  \displaystyle\sum_{\mu=1}^{N-1}(-1)^{\mu-1} \displaystyle\sum_{\substack{R\in \mathcal{B}_{N-\mu}(\mathcal{A})\\ |R|=N-\mu}} I^{|R| }_{R}+ (-1)^{N-1} S_{\cup A_i}~,\nonumber \\
\displaystyle\sum_{\mu=0}^{N-1}(-1)^{\mu} \displaystyle\sum_{\substack{R\in \mathcal{B}_{N-\mu}(\mathcal{A})\\ |R|=N-\mu}} I^{|R| }_{R} &=& (-1)^{N-1} S_{\cup A_i}~. 
\label{eq:tobeproved1}
\end{eqnarray}
We now prove eq.\eqref{eq:tobeproved1}. Using the definition of the multipartite information \eqref{eq:I_N_definition}, we can write
\begin{eqnarray}
\displaystyle\sum_{\substack{R\in \mathcal{B}_{m}(\mathcal{A})\\ |R|=m}} I^{m }_{R}  &=&  \displaystyle\sum_{\substack{R\in \mathcal{B}_{1}(\mathcal{A}) }} S_{R}  {N-1 \choose m-1} - \displaystyle\sum_{\substack{R\in \mathcal{B}_{2}(\mathcal{A}) }} S_{R}  {N-2 \choose m-2} + \cdots + (-1)^{m-1}  \displaystyle\sum_{\substack{R\in \mathcal{B}_{m}(\mathcal{A}) }} S_{R}  {N-2 \choose 0}~, \nonumber\\
&=& \displaystyle\sum_{\mu=1}^{m} (-1)^{\mu-1}  \displaystyle\sum_{\substack{R\in \mathcal{B}_{\mu}(\mathcal{A}) }} S_{R}  {N-\mu \choose m-1}~. 
\end{eqnarray}
Using this equation, we obtain
\begin{eqnarray}
&& \displaystyle\sum_{\mu=0}^{N-1}(-1)^{\mu} \displaystyle\sum_{\substack{R\in \mathcal{B}_{N-\mu}(\mathcal{A})\\ |R|=N-\mu}} I^{|R| }_{R} = (-1)^{N-1} \displaystyle\sum_{m=1}^{N} (-1)^{m-1} \displaystyle\sum_{\substack{R\in \mathcal{B}_{m}(\mathcal{A})\\ |R|=m}} I^{m }_{R}  \nonumber\\
&=&\displaystyle\sum_{\substack{R\in \mathcal{B}_{\mu}(\mathcal{A}) \\ \mu=1}}^{N-1} (-1)^{N-1} S_{R} \bigg[{N-\mu \choose N-\mu}- {N-\mu \choose N-\mu-1} +{N-\mu \choose N-\mu-2}-\cdots + (-1)^{N-\mu} {N-\mu \choose 0} \bigg]  +(-1)^{N-1} S_{\cup_i A_i}~,\nonumber\\
&=&  (-1)^{N-1} S_{\cup_i A_i}~,
\end{eqnarray}
where we have used the identity 
\begin{eqnarray*}
\bigg[{N-\mu \choose N-\mu}- {N-\mu \choose N-\mu-1} +{N-\mu \choose N-\mu-2}-\cdots + (-1)^{N-\mu} {N-\mu \choose 0} \bigg]=0~,~\forall N>\mu \in \mathbb{Z}~.
\end{eqnarray*}
Thus, we have proved the relation eq.\eqref{eq:tobeproved_master}, i.e., the expansion of the $N$-partite information in terms of various lower-order multipartite informations.

\section{Multipartite information constraint}
\label{ap:entanglemetn_constraint}

\par\noindent 
Following the discussion in Appendix \eqref{ap:holes_in_css}, eq.\eqref{eq:ap_1246} and the fact that $I^N_{A_N^{n_h}}=0$, we obtain
\begin{eqnarray}
\bigg[\sum_{j} (-1)^{(N-1)-\mu_j} I^{\mu_j}_{\{\mathcal{A}_{\mu_j}\}}\bigg] &+& \bigg[ (-1)^{(N-1)-2} \sum_{\{A_i,A_j\}\in \{M\}} I^2_{A_i,A_j}\bigg]+ (-1)^{N} \bigg[\displaystyle\sum_{i} S_{A_i} - S_{\cup A_i}  \bigg]=0 ~, ~~ \nonumber\\
\Rightarrow \sum_{j} |I^{\mu_j}_{\{\mathcal{A}_{\mu_j}\}}| &=&- \bigg[  \sum_{\{A_i,A_j\}\in \{M\}} I^2_{A_i,A_j}\bigg]+  \bigg[\displaystyle\sum_{i} S_{A_i} - S_{\cup A_i}  \bigg] ~, ~~ \nonumber\\
&=& d_{nn}\log\mathcal{D}+\bigg[ -N\log\mathcal{D}+(n_h+1)\log\mathcal{D} \bigg]=\bigg[d_{nn}-N+n_h+1\bigg]\log\mathcal{D}~, ~\nonumber\\
&=& 2n_h \log\mathcal{D}~. ~ \nonumber
\end{eqnarray}
The above result shows the dependence of the sum $\sum_{j} |I^{\mu_j}_{\{\mathcal{A}_{\mu_j}\}}|$ on the number of holes ($n_{h}$) of the CSS. Thus, we again find evidence for the dependence of the multipartite information measure of a topologically ordered ground state on the topology of the CSS. This can also be proved easily using Appendix \eqref{ap:simpleannular_euler}. For each closed loop, one obtains $|I^{\mu_j}_{\{A_j\}}|=\chi S_{topo} $. Thus, the total contribution arising from $n_{h}$ number of holes is simply
\begin{eqnarray}
\displaystyle\sum_{j=1}^{n_h} |I^{\mu_j}_{\{A_j\}}|=\chi n_h S_{topo}~.
\end{eqnarray}
\end{widetext}

\bibliography{TEE_Bibliography}

\end{document}